\documentclass[preprint,11pt]{elsarticle}
\usepackage{}

\usepackage{amsfonts}
\usepackage{amsmath}
\usepackage{amssymb}
\usepackage{latexsym}
\usepackage{graphicx}
\usepackage{ntheorem}

\usepackage{subcaption}
\usepackage[usenames]{color}

\usepackage{epsfig}

\input amssym.def
\newsymbol\rtimes 226F
\newfont{\nset}{msbm10}

\newtheorem{theo}{Theorem}[section]

\newtheorem{proposition}[theo]{Proposition}
\newtheorem{lemma}[theo]{Lemma}
\newtheorem{definition}[theo]{Definition}

\theoremstyle{empty}
\theorembodyfont{\rmfamily}
\newtheorem{refproof}{theorem}

\journal{Theoretical Computer Science}

\begin{document}

\begin{frontmatter}

\title{Maximum matchings in scale-free networks with identical degree distribution}

\author[lable1,label2]{Huan Li}

\author[lable1,label2]{Zhongzhi Zhang}
\ead{zhangzz@fudan.edu.cn}

\address[lable1]{School of Computer Science, Fudan
University, Shanghai 200433, China}
\address[label2]{Shanghai Key Laboratory of Intelligent Information
Processing, Fudan University, Shanghai 200433, China}

\begin{abstract}
The size and number of maximum matchings in a network have found a large variety of applications in many fields. As a ubiquitous property of diverse real systems, power-law degree distribution was shown to have a profound influence on size of maximum matchings in scale-free networks, where the size of maximum matchings is small and a perfect matching often does not exist. In this paper, we study analytically the maximum matchings in two scale-free networks with identical degree sequence, and show that the first network has no perfect matchings, while the second one has many. For the first network, we determine explicitly the size of maximum matchings, and provide an exact recursive solution for the number of maximum matchings. For the second one, we design an orientation and prove that it is Pfaffian, based on which we derive a closed-form expression for the number of perfect matchings. Moreover, we demonstrate that the entropy for perfect matchings is equal to that corresponding to the extended Sierpi\'nski graph with the same average degree as both studied scale-free networks. Our results indicate that power-law degree distribution alone is not sufficient to characterize the size and number of maximum matchings in scale-free networks.
\end{abstract}


\end{frontmatter}

\section{Introduction}

A matching in a graph with $N$ vertices is a set of edges, where no two edges are incident to a common vertex. A maximum matching is a matching of maximum cardinality, with a perfect matching being a particular case containing $\frac{N}{2}$ edges. The size and number of maximum matchings have numerous applications in physics~\cite{Mo64}, chemistry~\cite{Vu11}, computer science~\cite{LoPl86}, among others. For example, in the context of structural controllability~\cite{LiSlBa11,BaGe15}, the minimum number of driving vertices to control the whole network and the possible configurations of driving vertices are closely related to the size and number of maximum matchings in a bipartite graph.


Due to the relevance of diverse aspects, it is of theoretical and practical importance to study the size and number of maximum matchings in networks, which is, however, computationally intractable~\cite{Pr99}. Valiant proved that enumerating perfect matchings in general graphs is formidable~\cite{Va79TCS,Va79SiamJComput}, it is \#P-complete even in bipartite graph~\cite{Va79SiamJComput}. Thus, it is of great interest to find specific graphs for which the maximum matching problem can be solved exactly~\cite{LoPl86}. In the past decades, the problems related maximum matchings have attracted considerable attention from the community of mathematics and theoretical computer science~\cite{KaSi81,Ga96,Un97,MaVa00,GaKTa01,Jm03,LiLi04,Er04,YaYeZh05,YaZh05,YaZh06,Ke06,ZdMe06,YaZh08,TeSt09,ChFrMe10,DaAlTa12,KoNaPaPi13,Yu13,Me16}.


A vast majority of previous works about maximum matchings focused on regular graphs or random graphs~\cite{KaSi81,ZdMe06,ChFrMe10,Yu13}, which cannot well describe realistic networks. Extensive empirical works~\cite{Ne03} indicated that most real networks exhibit the prominent scale-free property~\cite{BaAl99}, with their degree distribution following a power law form. It has been shown that power law behavior has a strong effect on the properties of maximum matchings in a scale-free network~\cite{LiSlBa11}. For example, in the Barab\'asi-Albert (BA) scale-free network~\cite{BaAl99}, a perfect matching is almost sure not to exist, and the size of a maximum matching is much less than half the number of vertices. The same phenomenon was also observed in a lot of real scale-free networks, which are far from being perfectly matched as the BA network. Then, an interesting question arises naturally: whether the power-law degree distribution is the only ingredient characterizing maximum matchings in scale-free networks?



In order to answer the above-raised problem, in this paper, we present an analytical study of maximum matchings in two scale-free networks with identical degree distribution~\cite{ZhZhZoChGu09}, and show that the first network has no perfect matchings, while the second network has many. For the first network, we derive an exact expression for the size of maximum matchings, and provide a recursive solution to the number of maximum matchings. For the second network, by employing the Pfaffian method proposed independently by Kasteleyn~\cite{Ka61}, Fisher and Temperley~\cite{TeFi61}, we construct a Pfaffian orientation of the network. On the basis of Pfaffian orientation, we determine the number of perfect matchings  as well as its entropy, which is proved equal to that corresponding to the extended Sierpi\'nski graph~\cite{KlMo05}. Our findings suggest that the power-law degree distribution by itself cannot determine the properties of maximum matchings in scale-free networks.

%
%
%
%

\section{Preliminaries}

In this section, we introduce some useful notations and results that will be applied in the sequel.

\subsection{Graph and operation}

Let $\mathcal{G}=(\mathcal{V}(\mathcal{G}),\,\mathcal{E}(\mathcal{G}))$ be a graph with $N$ vertices and $E$ edges, where $\mathcal{V}(\mathcal{G})$ is the vertex set $\{v_1,v_2,\cdots,v_N\}$ and $\mathcal{E}(\mathcal{G})$ is the edge set. In this paper, all graphs considered are simple graphs without loops and parallel edges, having an even number of vertices. Throughout the paper, the two terms graph and network are used indistinctly.

Let $e=\{u,v\} \in \mathcal{E}(\mathcal{G})$ be an edge in  $\mathcal{G}$. We say that the edge $e$ is subdivided if we insert a new vertex $w$ between them, that is, edge $e$ is replaced by a path $u-w-v$ of length $2$. 
The subdivision graph $B(\mathcal{G})$ of a graph $\mathcal{G}$ is the graph obtaining from $\mathcal{G}$ by performing  the subdivision operation on each edge in $\mathcal{E}(\mathcal{G})$.

The line graph $L(\mathcal{G})$ of a graph $\mathcal{G}$ is the graph, where the vertex set is exactly the edge set $\mathcal{E}(\mathcal{G})$ of $\mathcal{G}$, and two vertices are adjacent if and only if their corresponding edges of $\mathcal{G}$ are connected to a common vertex in $\mathcal{V}(\mathcal{G})$.

The subdivided-line graph $\Gamma(\mathcal{G})$ of a graph $\mathcal{G}$
is the line graph of the subdivision graph of $\mathcal{G}$, i.e.,
$\Gamma(\mathcal{G}) = L(B(\mathcal{G}))$. We call $\Gamma$ the subdivided-line
graph operation. The $g$-iterated ($g\geq 1$) subdivided-line
graph $\Gamma^g(\mathcal{G})$ of $\mathcal{G}$ is the graph obtained from $\mathcal{G}$ by iteratively using the subdivided-line graph operation $g$ times.

\subsection{Structural properties of a graph}

For a network, the distance between two vertices is defined as the number of edges in the shortest path between them, and the average distance of the network is defined as the arithmetic average of distances over all pairs of vertices. The diameter of a network is the length of the shortest path between any pair of farthermost vertices in the network. A network is said to be small-world~\cite{WaSt98} if its average distance grows logarithmically with the number of vertices, or more slowly.

A random variable $x$ is said to follow a power law distribution if its probability density function $P(x)$ obeys the form $P(x) \sim x^{-\gamma}$. A network is scale-free~\cite{BaAl99} if its degree satisfies, at least asymptotically, a power law distribution $P(d) \sim d^{-\gamma}$. In realistic scale-free networks~\cite{Ne03}, the power exponent $\gamma$ of degree distribution typically lies between $2$ and $3$. Cumulative degree distribution $P_{\rm cum}(d)$ of a network is defined as $P_{\rm cum}(d)=\sum_{d'\geq d} P(d')$. For a scale-free network with  power law degree distribution $P(d)\sim d^{-\gamma}$, its cumulative degree distribution is also power law $P_{\rm cum}(d) \sim d^{-(\gamma-1)}$~\cite{Ne03}. 

In real networks there are nontrivial degree correlations among vertices~\cite{Ne02}. There are two measures characterizing degree correlations in a network. The first one is the average degree of the nearest neighbors for vertices with degree $d$ as a function of this degree value, denoted as $k_{\rm nn}(d)$~\cite{PaVaVe01}. When $k_{\rm nn}(d)$ is an increasing function of $d$, it means that vertices have a tendency to connect to vertices with a similar or larger degree. In this case, the network is called assortative. For example, the small-world Farey graph~\cite{ZhCo11,YiZhLiCh15} is assortative. In contrast, if $k_{\rm nn}(d)$ is a decreasing function of $d$, which implies that vertices of large degree are likely to be connected to vertices with small degree, then the network is said to be disassortative. And if $k_{\rm nn}(d)={\rm const}$, the network is uncorrelated.

The other quantity describing degree correlations is Pearson correlation coefficient of the degrees of the endvertices of all edges~\cite{Ne02}.
For a general graph $\mathcal{G}=(\mathcal{V}(\mathcal{G}),\,\mathcal{E}(\mathcal{G}))$, this coefficient is defined as
\begin{equation}\label{rNe}
r(\mathcal{G})=\frac{E \displaystyle\sum_{i=1}^E j_i \,k_i -\left[ \displaystyle\sum_{i=1}^E \frac{1}{2} (j_i+ k_i)\right]^2 }{E \displaystyle \sum_{i=1}^E \frac{1}{2} \left(j_i^2+ k_i^2\right) -\left[ \displaystyle\sum_{i=1}^E \frac{1}{2}(j_i+ k_i)\right]^2 }\,,
\end{equation}
where $j_i$, $k_i$ are the degrees of the two endvertices of the $i$th edge, with $i=1,2,\cdots, E$. The Pearson correlation coefficient is in the range
$-1\leq r(\mathcal{G}) \leq 1$. Network $\mathcal{G}$ is uncorrelated, if $r(\mathcal{G})=0$; $\mathcal{G}$ is disassortative, if $r(\mathcal{G})<0$; and $\mathcal{G}$ is assortative, if $r(\mathcal{G})>0$.


A network $\mathcal{G}$ is fractal if it has a finite fractal dimension, otherwise it is non-fractal~\cite{ SoHaMa05}. In general, the fractal dimension of $\mathcal{G}$ can be computed by a box-covering approach~\cite{SoGaHaMa07}.   {A box of size  $l_{B}$ is a vertex set such that all distances between any pair of vertices in the box are less than $l_{B}$.  We use boxes of  size $l_{B}$ to cover all vertices in $\mathcal{G}$, and let $N_B$ be minimum possible number of boxes  required to cover all vertices in $\mathcal{G}$.   If   $N_B$ satisfies $N_B \sim l_{B}^{-d_B}$ with $0<d_B<\infty$, then $\mathcal{G}$ is fractal with its fractal dimension being $d_B$.} Self-similarity of a network~\cite{KiGoKaKi07} refers to the scale invariance of the degree distribution under coarse-graining with various box sizes $l_{B}$ as well as under the iterative operations of coarse-graining with fixed $l_{B}$. Intuitively, a self-similar network resembles a part of itself. Notice that fractality and self-similarity do not always imply each other. A fractal network is self-similar, while a self-similar network may be not fractal.



\subsection{Matchings in a graph}

A matching in $\mathcal{G}$ is a subset $M \subseteq \mathcal{E}(\mathcal{G})$ such that no two edges in $M$ have a vertex in common. The size of a matching $M$ is the number of edges in $M$. A matching of the largest possible size is called a maximum matching. Matching number of $\mathcal{G}$ is defined as the cardinality of any maximum matching in $\mathcal{G}$. A vertex incident with an edge in  $M$ is said to be covered by  $M$. A matching $M$ is said to be perfect if it covers  every vertex of $\mathcal{G}$. Obviously, any perfect matching is a maximum matching. We use $\psi(\mathcal{G})$ to denote the number of perfect matchings of $\mathcal{G}$. 

\begin{lemma}\cite{DoYaZh13}   \label{linegraph}
For a connected graph $\mathcal{G}$ with $N$ vertices and $E$ edges, where
$E$ is even, the number $\psi(L(\mathcal{G}))$ of perfect matchings in its line graph satisfies $\psi(L(\mathcal{G})) \geq 2^{E-N+1}$, where equality holds
if the degree of all vertices in $\mathcal{G}$ is less than or equal to 3.
\end{lemma}

{A path is called an elementary one if it touches each vertex one time at most. A  cycle is  elementary if it touches each vertex one time at most, except the  starting and ending vertices.} In this paper, all paths and cycles mentioned are elementary paths and elementary cycles, respectively. A cycle $C$ of $\mathcal{G}$ is nice if $\mathcal{G} \setminus C$ contains a perfect matching, where $\mathcal{G} \setminus C$ represents the induced subgraph of $\mathcal{G}$ obtained from $\mathcal{G}$ by removing all vertices of $C$ and the edges connected to them. Similarly, a path $P$ of $\mathcal{G}$ is nice if $\mathcal{G} \setminus P$ includes a perfect matching. Since the number of vertices of any graph considered in this paper is even, the length (number of edges) of every nice cycle is even, while the length of every nice path is odd.

Let $\mathcal{G}^{e}$ be an orientation of $\mathcal{G}$. Then the skew adjacency matrix of $\mathcal{G}^{e}$, denoted by $A(\mathcal{G}^{e})$, is defined as
\begin{equation}
A(\mathcal{G}^{e})=(a_{ij})_{N \times N}\,,
\end{equation}
where
\begin{equation}\label{Equ:Pre01}
a_{ij}=\begin{cases}
1, &(v_i,v_j) \in \mathcal{E}(\mathcal{G}^{e})\,, \\
-1, &(v_j,v_i) \in \mathcal{E}(\mathcal{G}^{e})\,, \\
0, &\text{otherwise}\,.
\end{cases}
\end{equation}

For a cycle $C$ with even length, we shall say $C$ is oddly (or evenly) oriented in $\mathcal{G}^{e}$ if $C$ contains an odd (or even) number of co-oriented edges when the cycle is traversed in either direction. Similarly, a path $P$ is said to be oddly (or evenly) oriented in $\mathcal{G}^{e}$ if it has an odd (or even) number of co-oriented edges when $P$ is traversed from its starting vertex to ending vertex.  $\mathcal{G}^{e}$ is a Pfaffian orientation of $\mathcal{G}$ if every nice cycle of $\mathcal{G}$ is oddly oriented in $\mathcal{G}^{e}$~\cite{LoPl86}.

If $\mathcal{G}^{e}$ is a Pfaffian orientation of network $\mathcal{G}$, then the number of perfect matchings of $\mathcal{G}$, $\psi(\mathcal{G})$, is equal to the square root of the determinant of $A(\mathcal{G}^{e})$:
\begin{lemma}\label{Lem:Pre01}\cite{LoPl86,Ka63}
Let $\mathcal{G}^{e}$ be a Pfaffian orientation of $\mathcal{G}$. Then
\begin{equation}\label{Equ:Pre01}
\psi(\mathcal{G}) =\sqrt{\det \left(A(\mathcal{G}^{e})\right)}\,.
\end{equation}
\end{lemma}

An interesting quantity related to perfect matchings is entropy. For a network $\mathcal{G}$ with sufficiently large number of vertices, the entropy for perfect matchings is defined as follows~\cite{BuPe93,WU200620}:
\begin{equation}\label{Equ:Pre02}
z(\mathcal{G})=\lim\limits_{N\to \infty} \frac{\ln \psi(\mathcal{G})}{ \frac{N}{2} } \,.
\end{equation}


After introducing related notations, in what follows we will study maximum matchings of two scale-free networks with the same degree sequence~\cite{ZhZhZoChGu09}: one is fractal and large-world, while the other is non-fractal and small-world. We will show that for both networks, the properties of their maximum matchings differ greatly.

\section{Maximum matchings in a fractal and scale-free network}

In this section, we study the size and number of maximum matchings for a fractal scale-free network.

\subsection{Construction and structural properties}

We first introduce the construction methods of the fractal network and study some of its structural properties.
The fractal scale-free network is generated in an iterative way~\cite{BeOs79}.
\begin{definition}\label{Def:LW01}
Let $\mathcal{F}_g=(\mathcal{V}(\mathcal{F}_g),\,\mathcal{E}(\mathcal{F}_g))$, $g \geq 1$, denote the fractal scale-free network after $g$ iterations, with $\mathcal{V}(\mathcal{F}_g)$ and $\mathcal{E}(\mathcal{F}_g)$ being the vertex set and the edge set, respectively. Then, $\mathcal{F}_g$ is constructed as follows:

For $g=1$, $\mathcal{F}_1$ is a quadrangle containing four vertices and edges.

For $g>1$, $\mathcal{F}_g$ is obtained from $\mathcal{F}_{g-1}$ by replacing each edge of $\mathcal{F}_{g-1}$ with a quadrangle on the right-hand side (rhs) of the arrow in Fig.~\ref{Fig.11}.
\end{definition}

\begin{figure}
\centering
\includegraphics[width=0.4\textwidth]{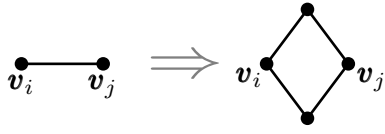}
\caption{Construction method for the fractal scale-free network. To obtain network of next iteration, each edge $(v_i,v_j)$ of current iteration is replaced by a quadrangle on the rhs of the arrow, with $v_i$ and $v_j$ being two diagonal vertices.}
\label{Fig.11}
\end{figure}

Figure~\ref{Fig.12} illustrates the construction process of the first several iterations.

\begin{figure}
\centering
\includegraphics[width=0.55\textwidth]{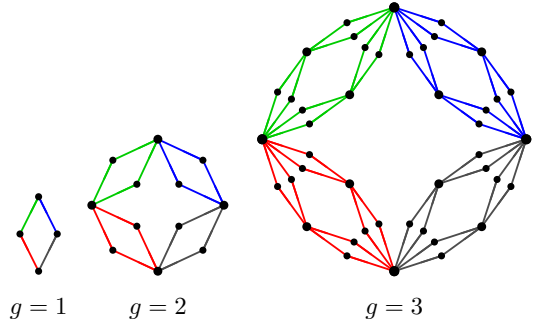}
\caption{The iteration process for the fractal scale-free network.}
\label{Fig.12}
\end{figure}

The fractal scale-free network is self-similar, which
can be easily seen from alternative construction approach~\cite{ZhZhZo07}. As will be shown below, in $\mathcal{F}_g$, $g \geq 1$, there are four vertices with the largest degree, which we call hub vertices. For the four hub vertices in $\mathcal{F}_1$, we label one pair of diagonal vertices as $v_1$ and $v_2$, and label the other pair of vertices as $v_3$ and $v_4$. Then, the fractal scale-free network can be created in another way as illustrated in Fig.~\ref{Fig.13}.

\begin{definition}\label{Def:LW02}
Given the network $\mathcal{F}_{g-1}=(\mathcal{V}(\mathcal{F}_{g-1}),\,\mathcal{E}(\mathcal{F}_{g-1}))$, $g > 1$, $\mathcal{F}_g=(\mathcal{V}(\mathcal{F}_g),\,\mathcal{E}(\mathcal{F}_g))$ is obtained by performing the following operations:

(i)  Merging four replicas of $\mathcal{F}_{g-1}$, denoted by $\mathcal{F}_{g-1}^{(i)}$, $i=1,2,3,4$, the four hub vertices of which are denoted by $v_k^{(i)}$, $k=1,2,3,4$, with $v_k^{(i)}$ in $\mathcal{F}_{g-1}^{(i)}$ corresponding to $v_k$ in $\mathcal{F}_{g-1}$.

(ii)  Identifying $v_1^{(1)}$ and $v_1^{(4)}$ (or $v_2^{(2)}$ and $v_1^{(3)}$, $v_2^{(1)}$ and $v_1^{(2)}$, $v_2^{(3)}$ and $v_2^{(4)}$) are as the hub vertex $v_1$ (or $v_2$,  $v_3$,  $v_4$) in $\mathcal{F}_g$.
\end{definition}


\begin{figure}
\centering
\includegraphics[width=0.55\textwidth]{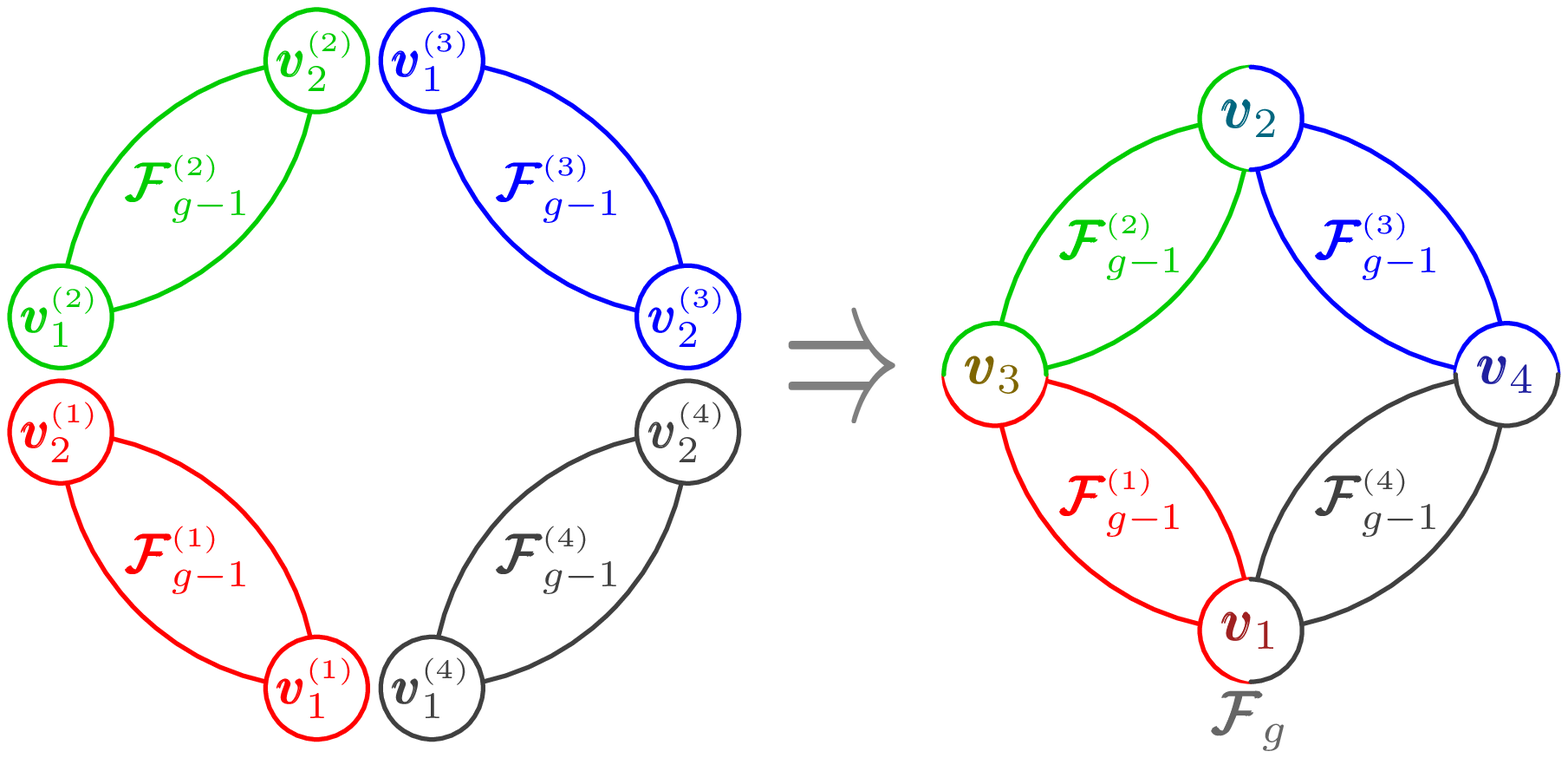}
\caption{Second approach for constructing the fractal scale-free network.} 
\label{Fig.13}
\end{figure}

Let $N_g$ and $E_g$ denote, respectively, the number of vertices and edges in $\mathcal{F}_g$.
By construction, $N_g$ and $E_g$ obey relations $N_g = 4N_{g-1}-4$ and $E_g=4 E_{g-1}$.
With the initial condition $N_1=E_1=4$, we have $N_g=\frac{2}{3}\left(4^{g}+2 \right)$ and $E_g=4^g$.

{According to the first construction, we can determine the degree for all vertices in  $\mathcal{F}_g$ and its distribution.} Let $L_v(g_i)$ denote the number of vertices created at iteration $g_i$. Then, $L_v(1)=4$ and $L_v(g_i)=2\times 4^{ g_i -1}$ for $ g_i >1$. In network $\mathcal{F}_g$, any two vertices generated at the same iteration have the same degree. Let $d_i(g)$ be the degree of a vertex in $\mathcal{F}_g$, which was created at iteration $g_i$. {After each iteration, the degree of each vertex doubles, implying $d_i(g)=2\,d_i(g-1)$, which together with $d_i(g_i)=2$  leads to $d_i(g)= 2^{g-g_i+1}$.} Thus, all possible degree of vertices in $\mathcal{F}_g$ is $2,2^2,2^3,\dots,2^{g-1},2^{g}$, and the number of  vertices with degree $\delta =2^{ g-g_i+1}$ is $L_v(g_i)$. From the degree sequence we can determine the cumulative degree distribution of $\mathcal{F}_g$.
\begin{proposition}\cite{ZhZhZo07}
The cumulative degree distribution of network $\mathcal{F}_g$  obeys a power law form $P_{\rm cum}(d)\sim d^{-2}$.
\end{proposition}
Thus, network $\mathcal{F}_g$ is scale-free with its power exponent $\gamma$ of degree distribution being 3.

\begin{proposition}\cite{HiBe06}
The average distance of $\mathcal{F}_g$ is
$$
\mu(\mathcal{F}_g) = \frac{22\times 2^g \times 16^g+8^g(21g+42)+27\times 4^g+98\times 2^g}{42\times16^g+105\times 4^g+42}.
$$
\end{proposition}
For large $g$, $\mu(\mathcal{F}_g) \sim \frac{11}{21}2^g$.
On the other hand, when $g$ is very large, $N_g \sim \frac{2}{3}\cdot 4^g$. Thus, $\mu(\mathcal{F}_g)$ grows as a square root of the number of vertices, implying that the network is ``large-world",  instead of small-world. Note that although most real networks are small-world, there exist some ``large-world" networks, e.g. global network of avian influenza outbreaks \cite{SmXuZhZhSuLu08}.

In addition, the network $\mathcal{F}_g$ is fractal and disassortative.
\begin{proposition}\cite{HiBe06}
The network $\mathcal{F}_g$ is fractal with the fractal dimension being 2.
\end{proposition}

\begin{proposition}\cite{ZhZhZo07}
In network $\mathcal{F}_g$, $g\geq 1$, the average degree of the neighboring vertices for vertices with degree $d$ is
$$k_{\rm nn} (d) =
\begin{cases}
2g, & d = 2,\\
2, & d > 2.
\end{cases}
$$
\label{corrFt0}
\end{proposition}
Thus, network $\mathcal{F}_g$ is disassortative, which can also be seen from its  Pearson correlation coefficient.

\begin{proposition}
The Pearson  correlation coefficient of the network $\mathcal{F}_g$, $g\geq 1$, is
$$r(\mathcal{F}_g)= \frac{(g-1)^2}{-3\times2^g+g^2+2g+3}.$$
\label{corrFt}
\end{proposition}
\begin{proof}
We first calculate the following three summations over all $E_g$ edges in $\mathcal{F}_g$.
\begin{align*}
\sum \limits_{m=1}^{E_g} j_m k_m  &=  \frac{1}{2} \sum \limits_{m=1}^{E_g} (j_m k_m + k_m j_m)
 =  \frac{1}{2} \sum \limits_{i=1}^{N_g}\left[ d_i(g)\sum\limits_{(i,j)\in \mathcal{E}(\mathcal{F}_g)} d_j(g) \right ], \\
&= \frac{1}{2} \sum \limits_{g_i=1}^{g} L_v(g_i) d^2_i(g) k_{\rm nn}(d_i(g))
 =  g\cdot 4^{g+1},
\end{align*}
\begin{equation*}
\sum \limits_{m=1}^{E_g} (j_m  + k_m)  = \sum \limits_{i=1}^{N_g}d_i^2(g)
 =  \sum \limits_{g_i=1}^{g} L_v(g_i) d^2_i(g)
 =  (g+1)2^{2g+1},
\end{equation*}
\begin{equation*}
\sum \limits_{m=1}^{E_g} (j_m^2  + k_m^2) = \sum \limits_{i=1}^{N_g}d_i^3(g)
 =  \sum \limits_{g_i=1}^{g} L_v(g_i) d^3_i(g)
 =  2^{2g+1}(3\times 2^g-2).
\end{equation*}
Inserting these results into Eq.~\eqref{rNe} and considering $E_g=4^g$ yields the result. 
\end{proof}



\subsection{Size of maximum matchings}

Although for a general graph, the size of maximum matchings is not easy to determine, for the fractal scale-free network $\mathcal{F}_g$, we can obtain it by using its self-similar structure.
\begin{theo}\label{theomm}
The matching number of network $\mathcal{F}_g$ is $\frac{4^g+8}{6}$.
\end{theo}
\begin{proof}
In order to determine the size of maximum  matchings of network $\mathcal{F}_g$, denoted by $c_g$, we define some useful quantities. Let $a_g$ be the size of maximum matchings of $\mathcal{F}_g\setminus \{v_1,v_2\}$, and let
$b_g$ denote the size of maximum matchings of $\mathcal{F}_g\setminus \{v_1\}$, which equals the size of maximum matchings of  $\mathcal{F}_g\setminus \{v_2\}$. We now determine the three quantities $a_g$, $b_g$, and $c_g$ by using the self-similar architecture of the network.

Figs.~\ref{ageq},~\ref{bgeq},~\ref{cgeq} show, respectively, all the possible configurations of maximum  matchings of network $\mathcal{F}_{g+1}\setminus \{v_1,v_2\}$, $\mathcal{F}_{g+1}\setminus \{v_1\}$, and $\mathcal{F}_{g+1}$. Note that in Figs.~\ref{ageq},~\ref{bgeq},~\ref{cgeq}, only the hub vertices are shown explicitly, with solid line hubs being covered and dotted line hubs being vacant. From Figs.~\ref{ageq},~\ref{bgeq},~\ref{cgeq}, we can establish recursive relations for $a_g$, $b_g$, and $c_g$, given by
\begin{align} \label{mmeq}
a_{g+1} &= 2a_g+2b_g\,, \notag\\
b_{g+1} &= {\rm max}\{2a_g+b_g+c_g ,a_g+3b_g \}\,, \notag\\
c_{g+1} &= {\rm max}\{2a_g+2c_g,4b_g,a_g+2b_g+c_g \}\,.
\end{align}
With initial condition $a_1=0$, $b_1=1$, and $c_1=2$, the above equations are solved to obtain $a_g = \frac{4^g-4}{6}$, $b_g = \frac{4^g+2}{6}$ and $c_g = \frac{4^g+8}{6}$.
\end{proof}


\begin{figure}
    \centering
    \begin{subfigure}[b]{0.2\textwidth}
        \includegraphics[width=\textwidth]{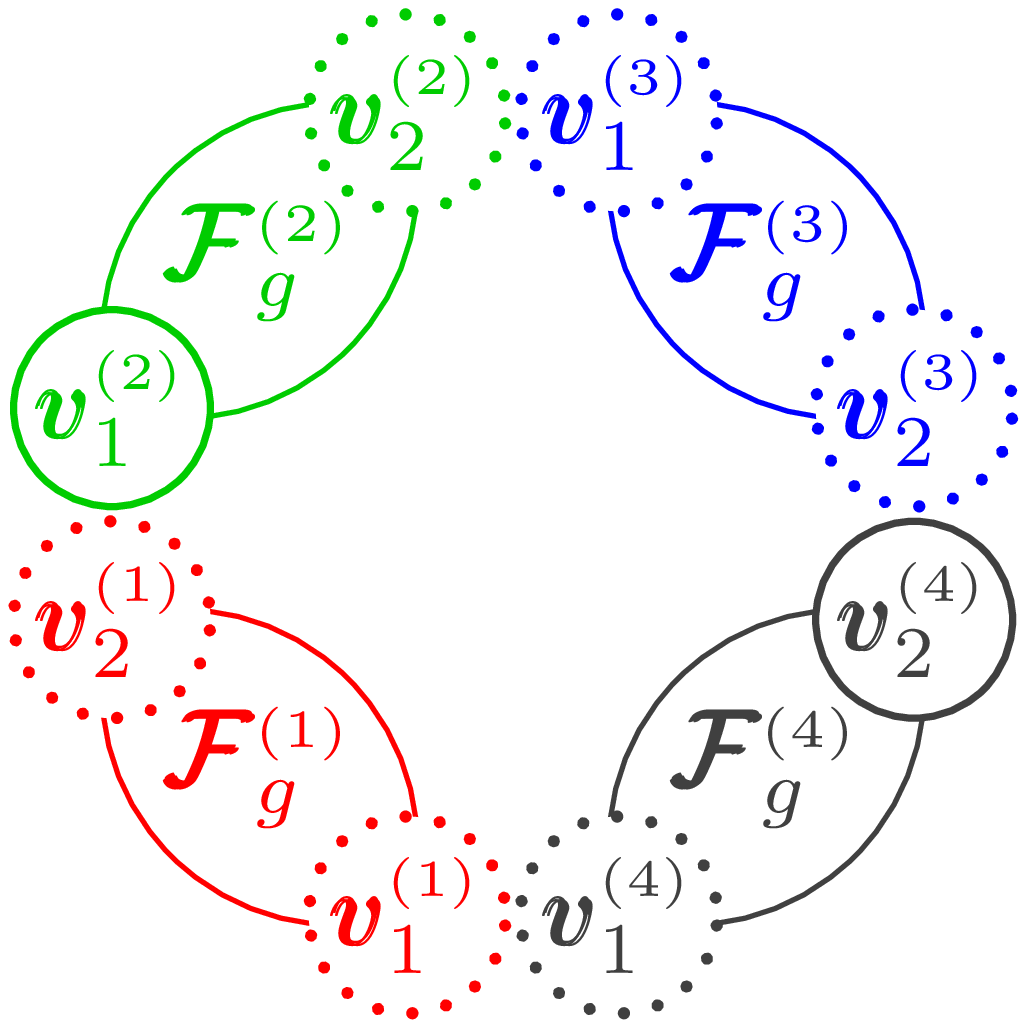}
    \end{subfigure}
    \begin{subfigure}[b]{0.2\textwidth}
        \includegraphics[width=\textwidth]{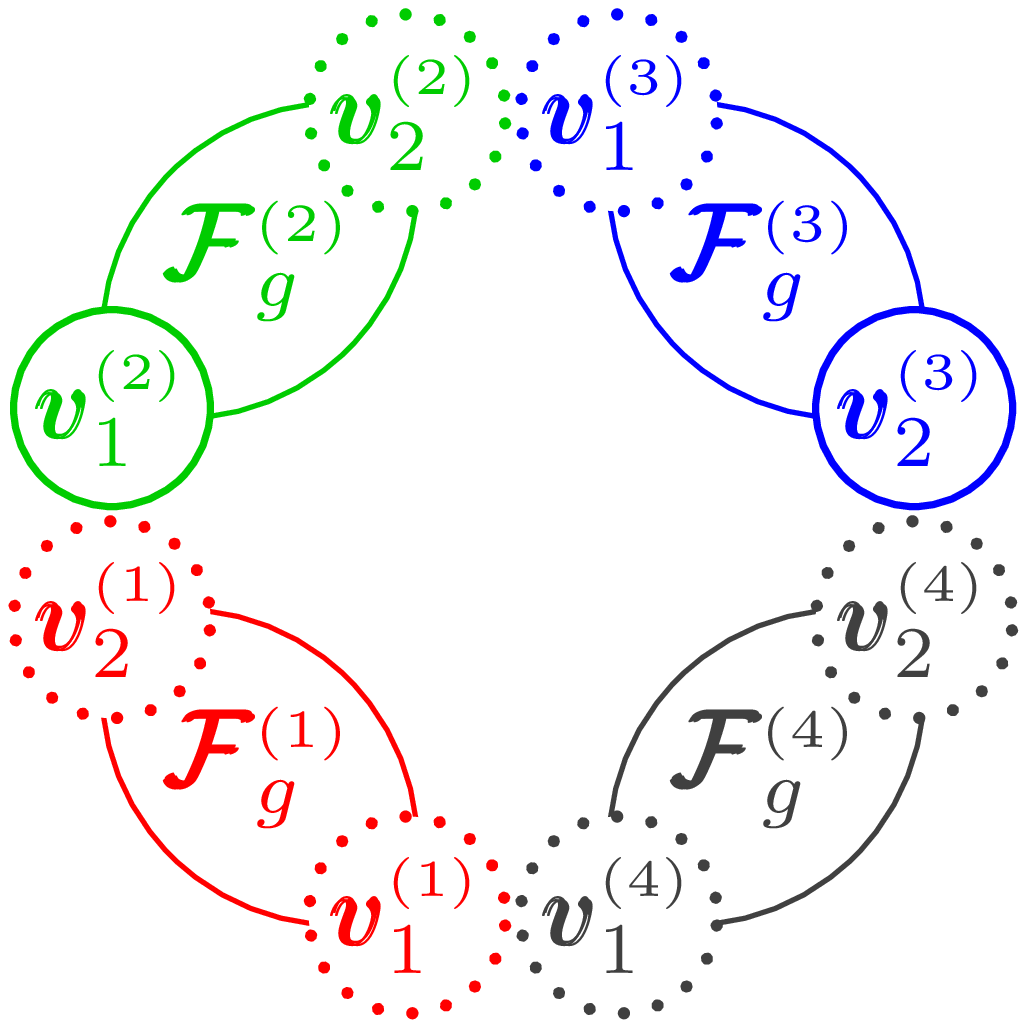}
    \end{subfigure}
    \begin{subfigure}[b]{0.2\textwidth}
        \includegraphics[width=\textwidth]{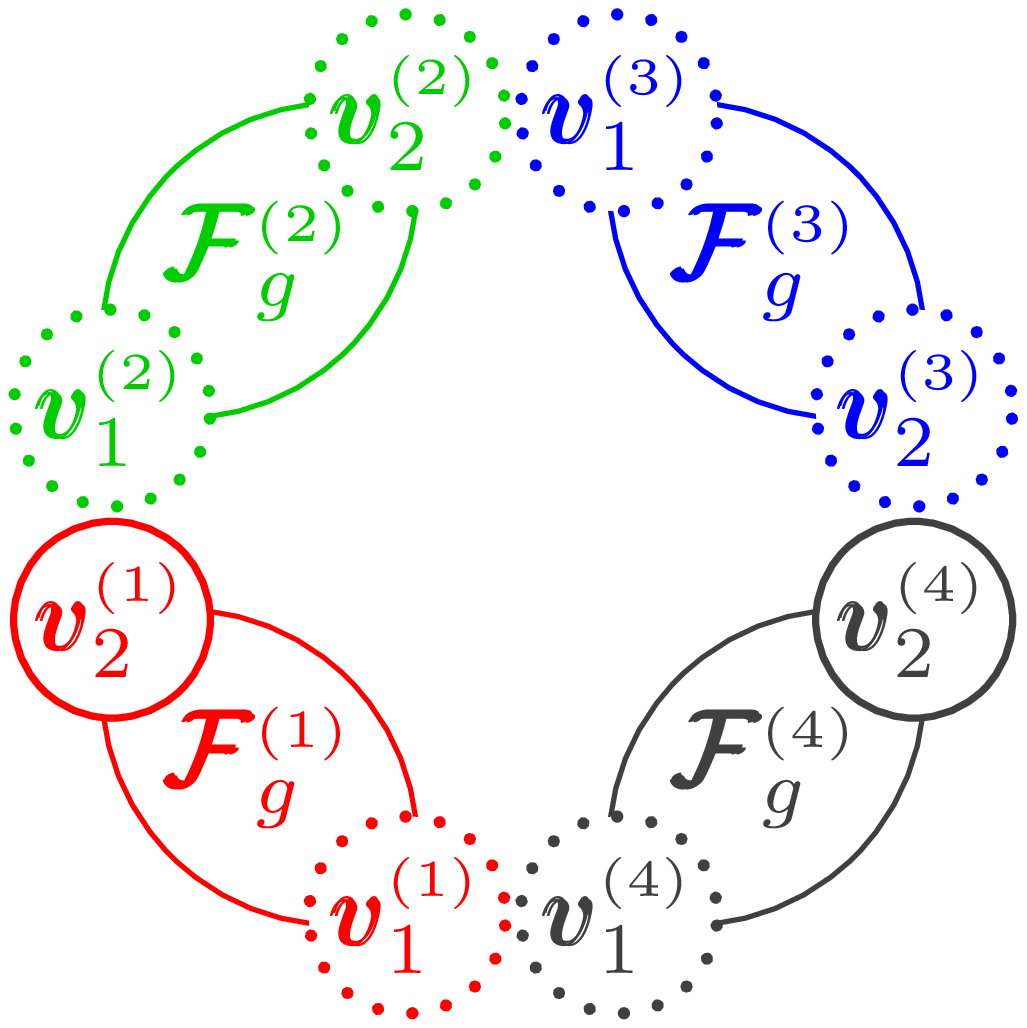}
    \end{subfigure}
    \begin{subfigure}[b]{0.2\textwidth}
        \includegraphics[width=\textwidth]{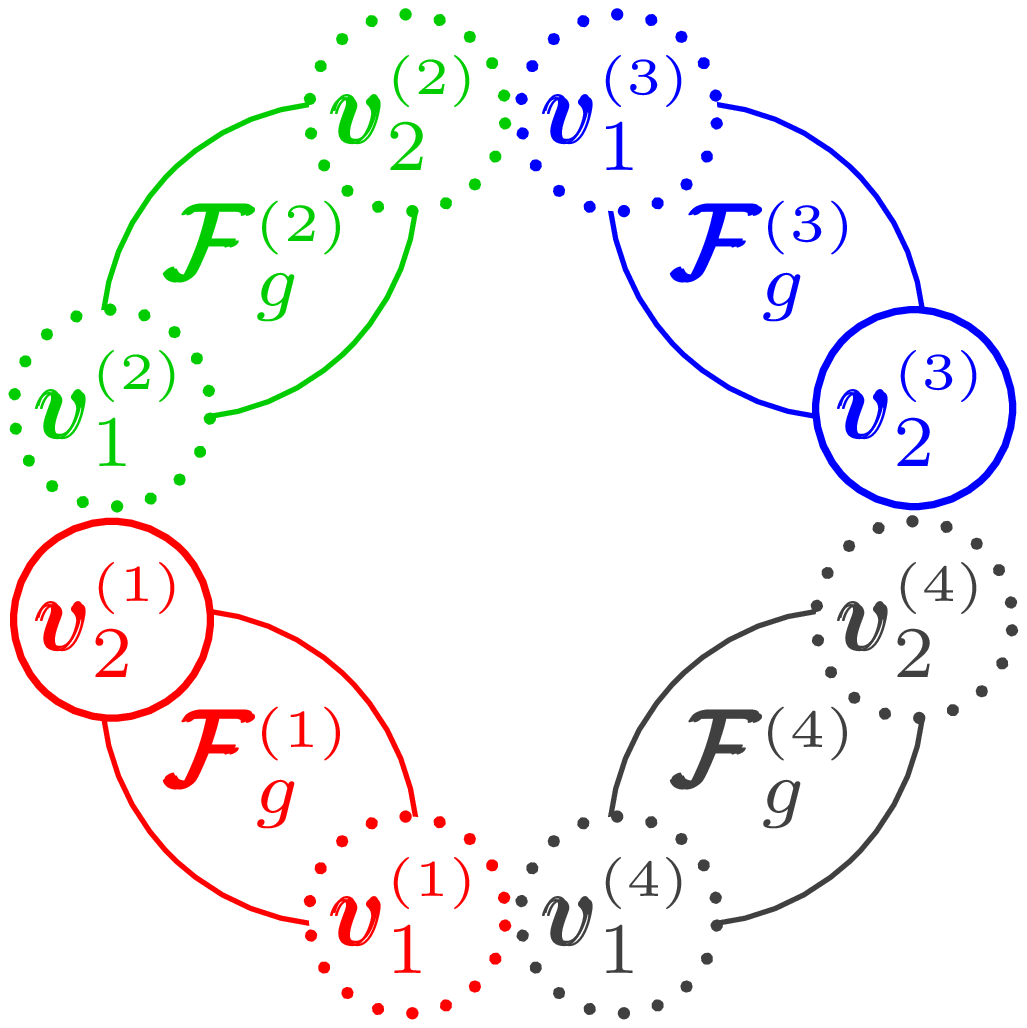}
    \end{subfigure}
    \caption{Illustration of all possible configurations of maximum matchings for $\mathcal{F}_{g+1} \setminus \{v_1,v_2\}$. {For each maximum matching belonging to the leftmost  configuration, the number of edges in $\mathcal{F}^{(1)}_{g}$, $\mathcal{F}^{(2)}_{g}$, $\mathcal{F}^{(3)}_{g}$,  and $\mathcal{F}^{(4)}_{g}$, is $a_{g}$, $b_{g}$, $a_{g}$, and $b_{g}$, respectively. Thus, the size of each maximum matching in the leftmost configuration is  $2a_{g}+2b_{g}$. Analogously, the size of each maximum matching in the other three configurations is also $2a_{g}+2b_{g}$.}}
  \label{ageq}
  \end{figure}

\begin{figure}
\begin{subfigure}[b]{\textwidth}
    \centering
    \begin{subfigure}[b]{0.2\textwidth}
        \includegraphics[width=\textwidth]{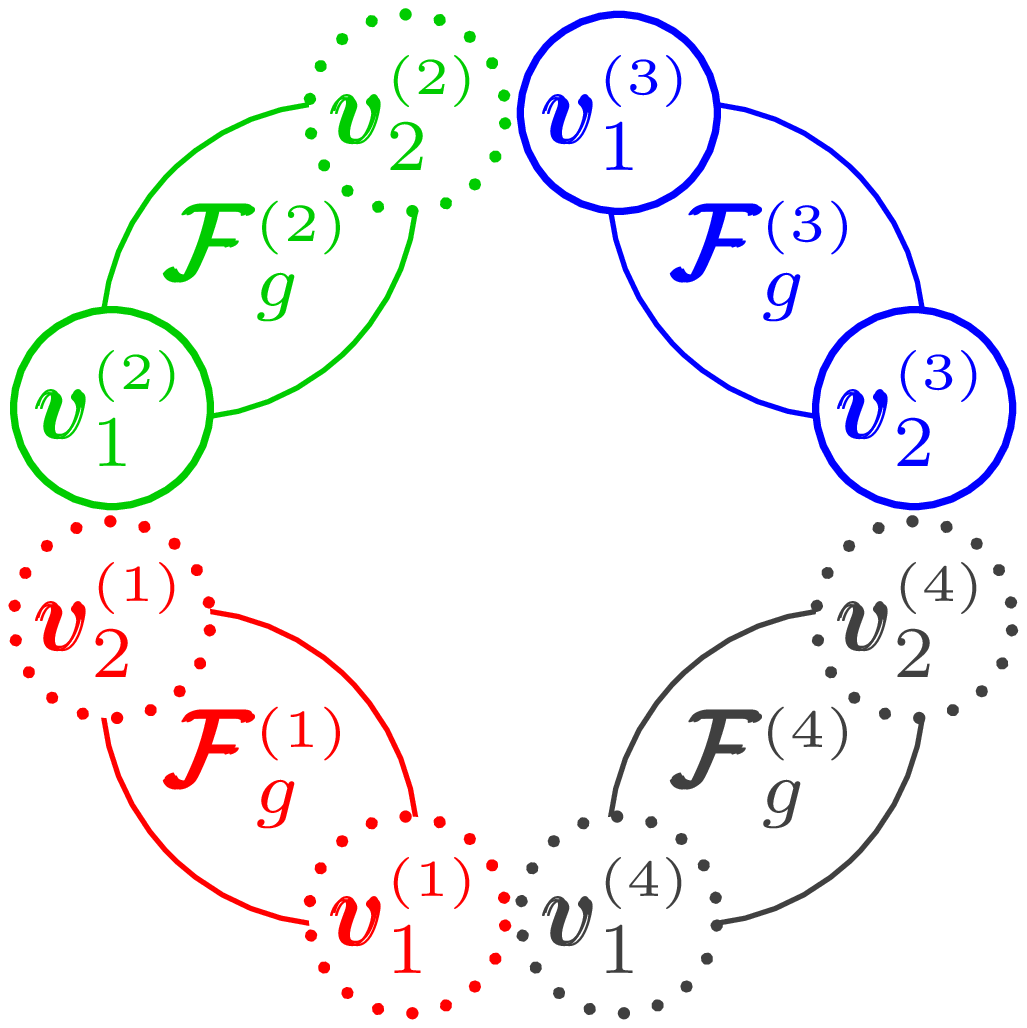}
    \end{subfigure}
    \begin{subfigure}[b]{0.2\textwidth}
        \includegraphics[width=\textwidth]{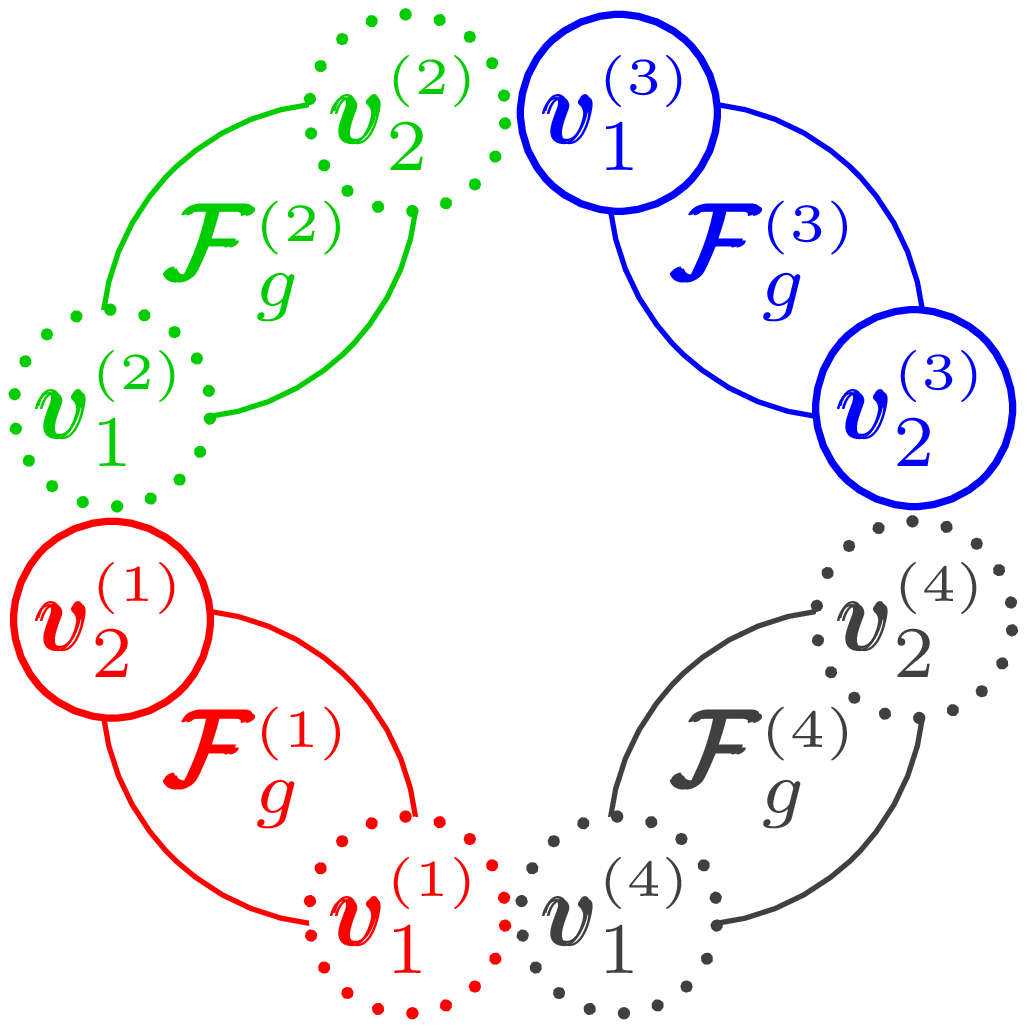}
    \end{subfigure}
    \begin{subfigure}[b]{0.2\textwidth}
        \includegraphics[width=\textwidth]{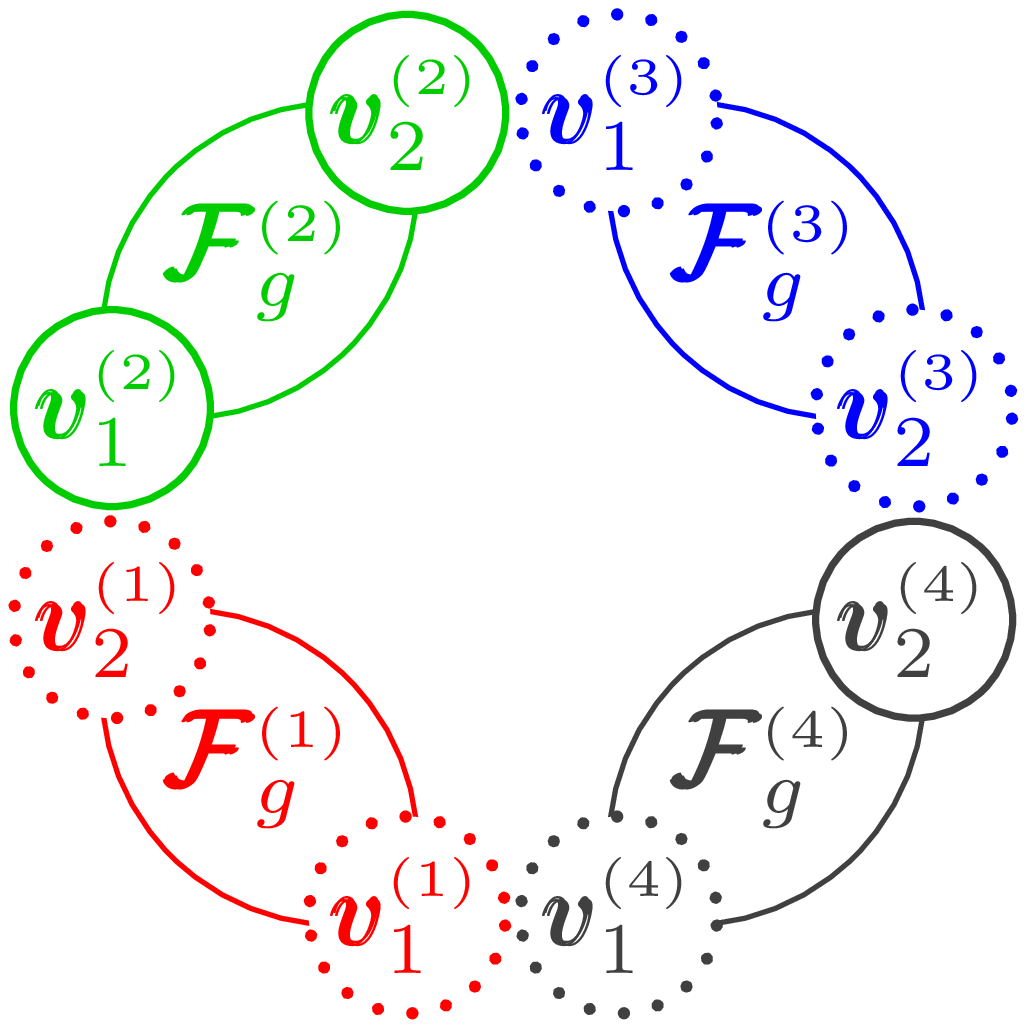}
    \end{subfigure}
    \begin{subfigure}[b]{0.2\textwidth}
    	 \includegraphics[width=\textwidth]{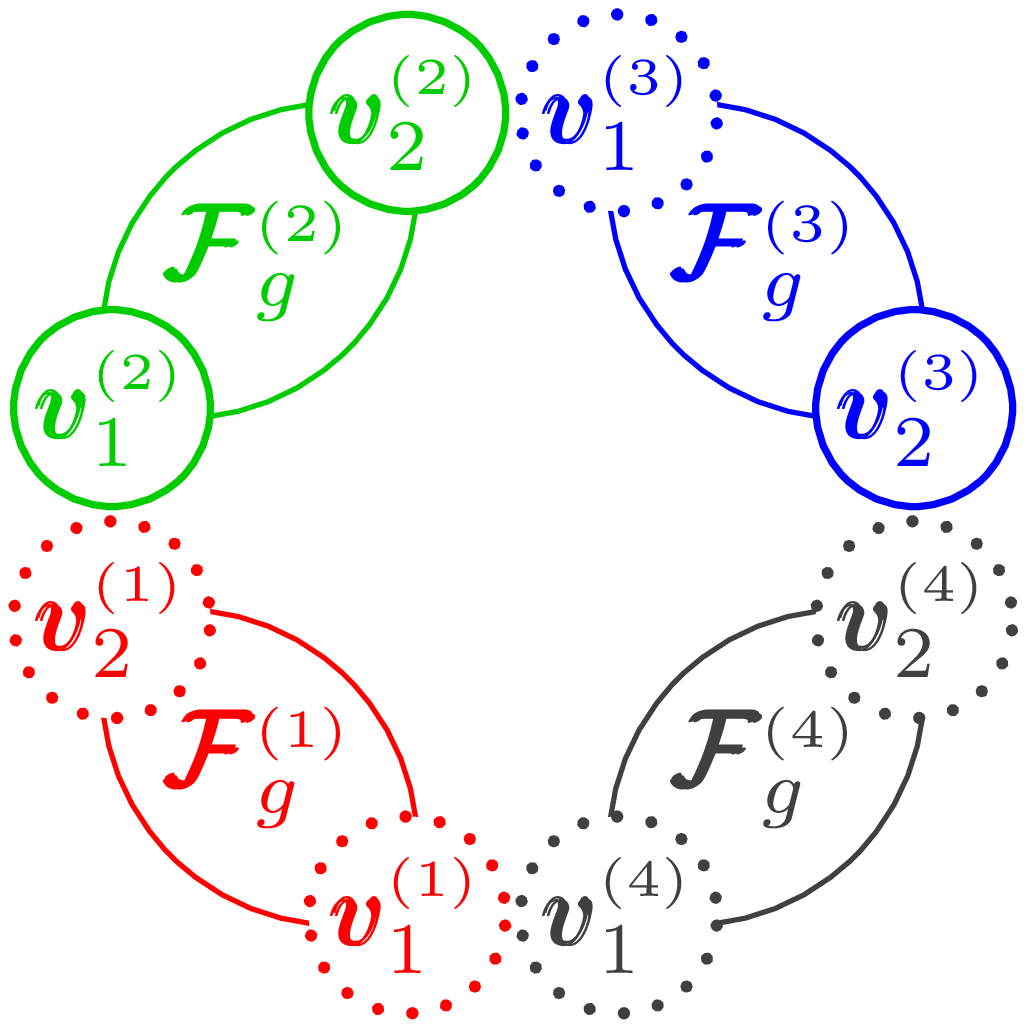}
    \end{subfigure}
    \caption{Configurations of maximum matchings for $\mathcal{F}_{g+1}\setminus \{v_1\}$ with size \(2a_g+b_g+c_g\).}
\end{subfigure}
\vfill
\null
\begin{subfigure}[b]{\textwidth}
    \centering
    \begin{subfigure}[b]{0.2\textwidth}
        \includegraphics[width=\textwidth]{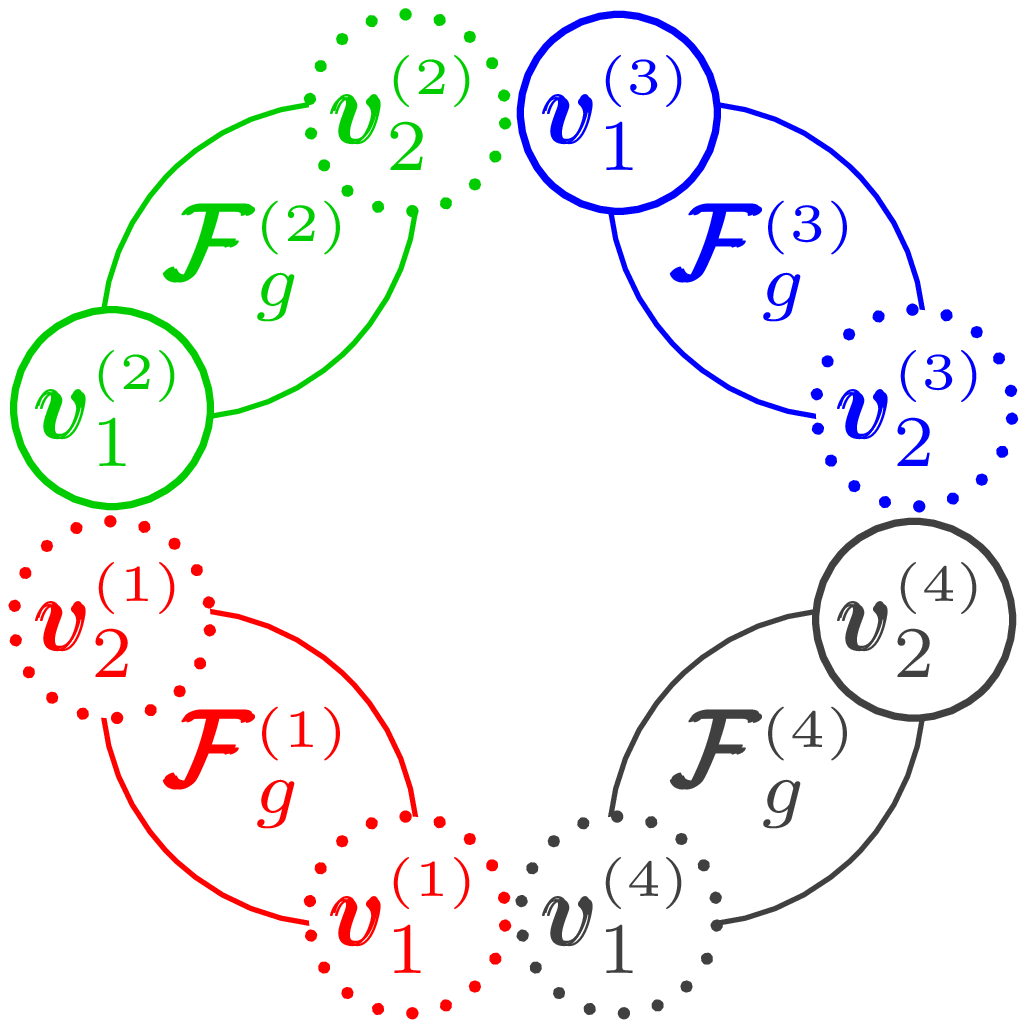}
    \end{subfigure}
    \begin{subfigure}[b]{0.2\textwidth}
        \includegraphics[width=\textwidth]{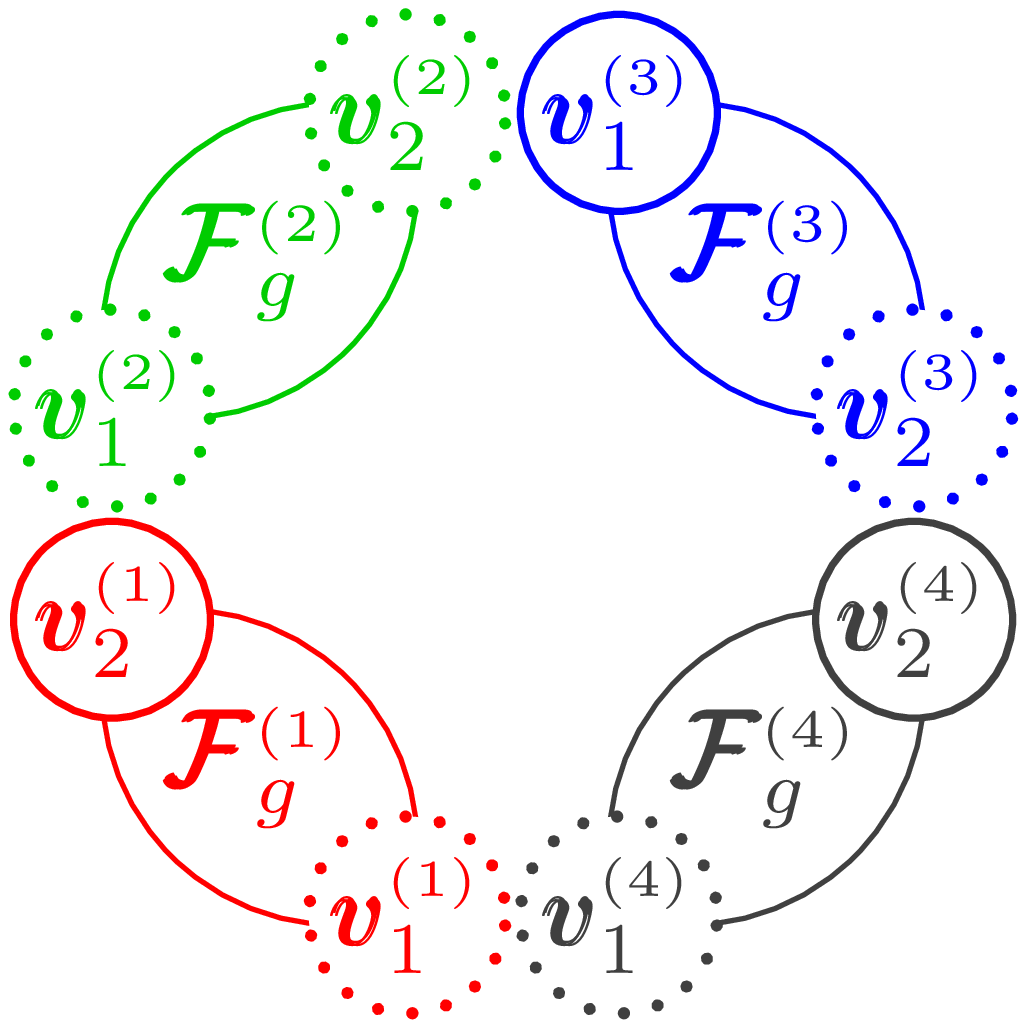}
    \end{subfigure}
    \begin{subfigure}[b]{0.2\textwidth}
        \includegraphics[width=\textwidth]{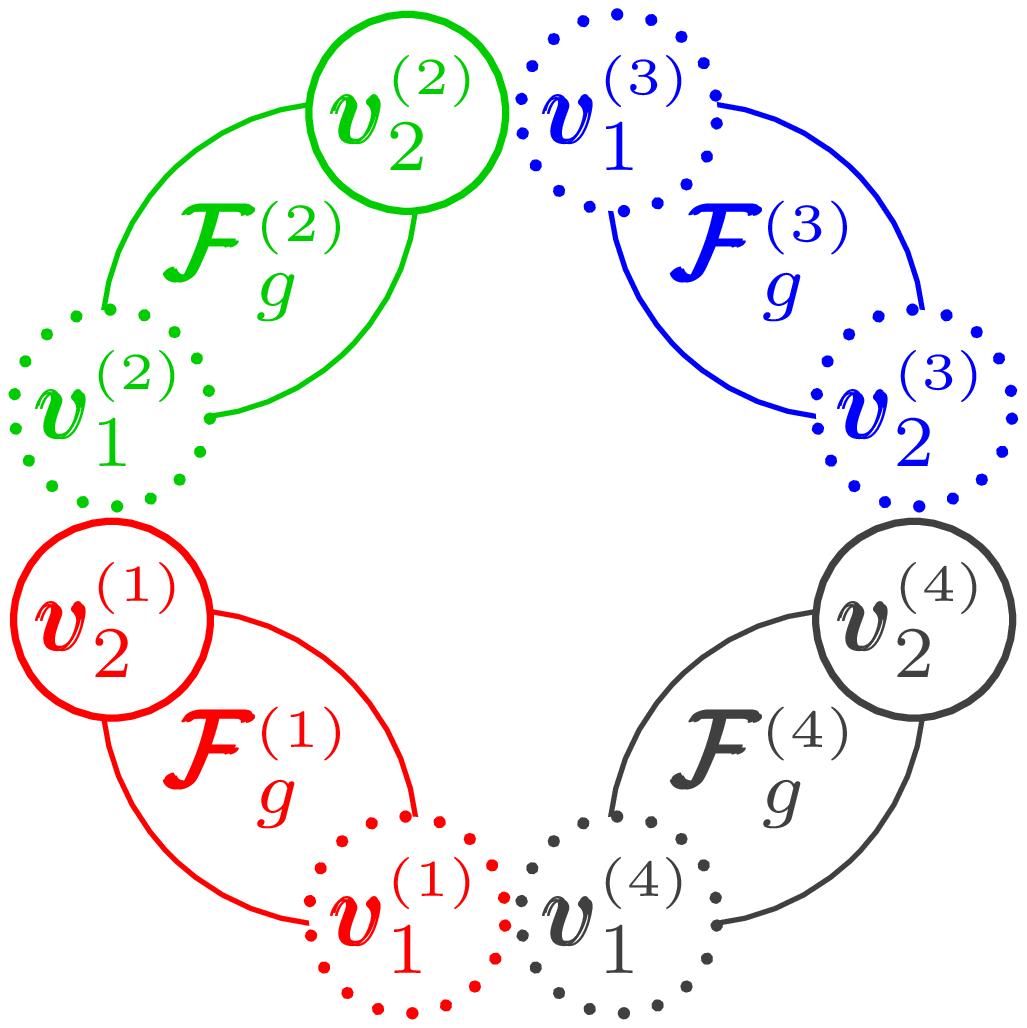}
    \end{subfigure}
    \begin{subfigure}[b]{0.2\textwidth}
        \includegraphics[width=\textwidth]{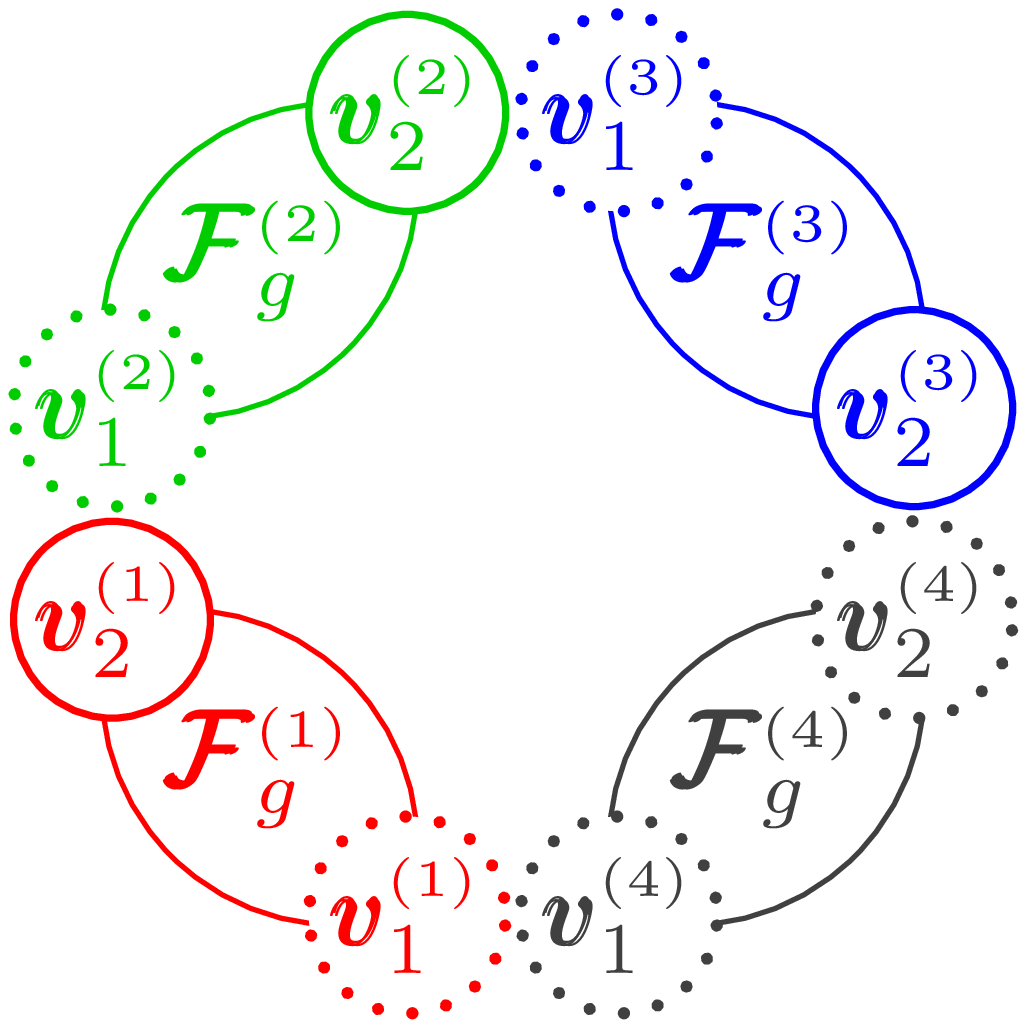}
    \end{subfigure}
    \caption{Configurations of maximum matchings for $\mathcal{F}_{g+1}\setminus \{v_1\}$ with size \(a_g+3b_g\).}
\end{subfigure}
\caption{Illustration of all possible configurations of maximum matchings for $\mathcal{F}_{g+1}\setminus \{v_1\}$. }
\label{bgeq}
\end{figure}

\begin{figure}
\centering
\begin{subfigure}[b]{0.49\textwidth}
	\centering
	\begin{subfigure}[b]{0.4\textwidth}
		\includegraphics[width=\textwidth]{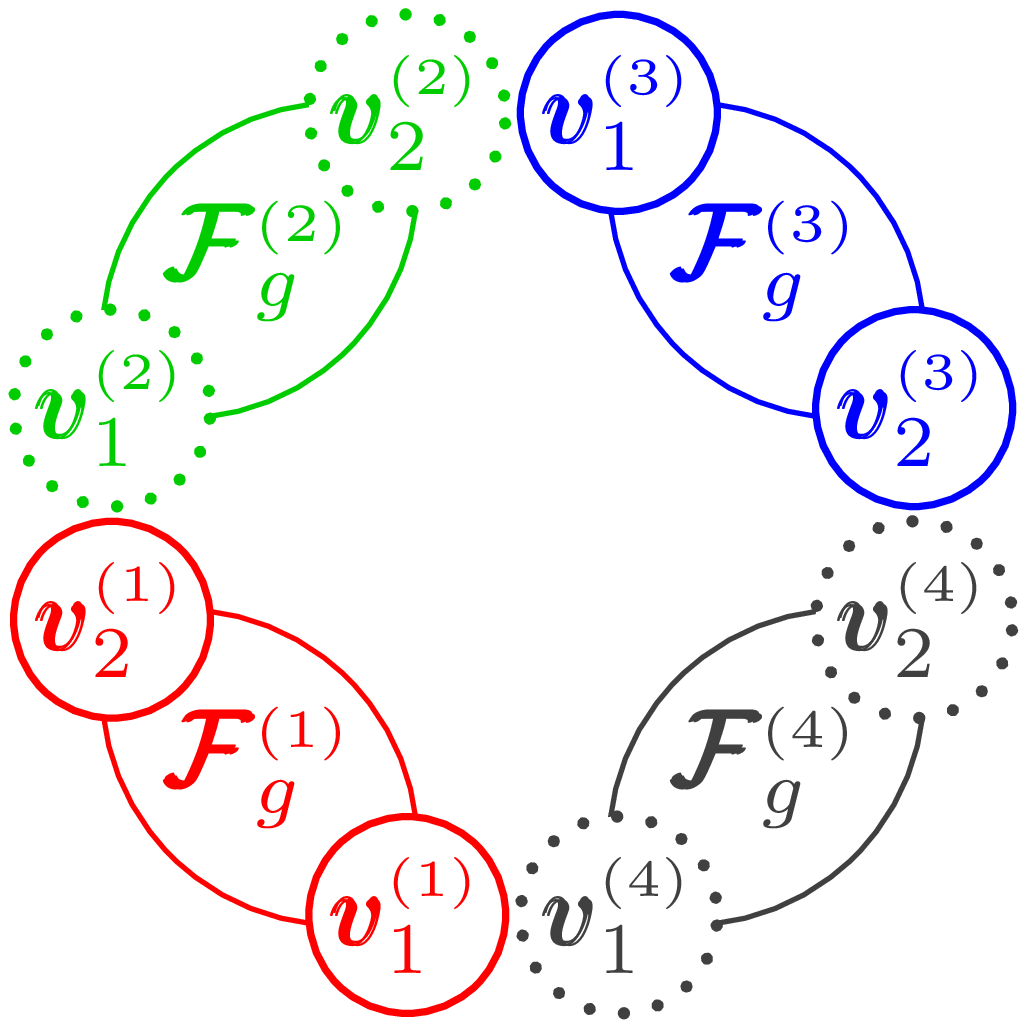}
	\end{subfigure}
	\begin{subfigure}[b]{0.4\textwidth}
		\includegraphics[width=\textwidth]{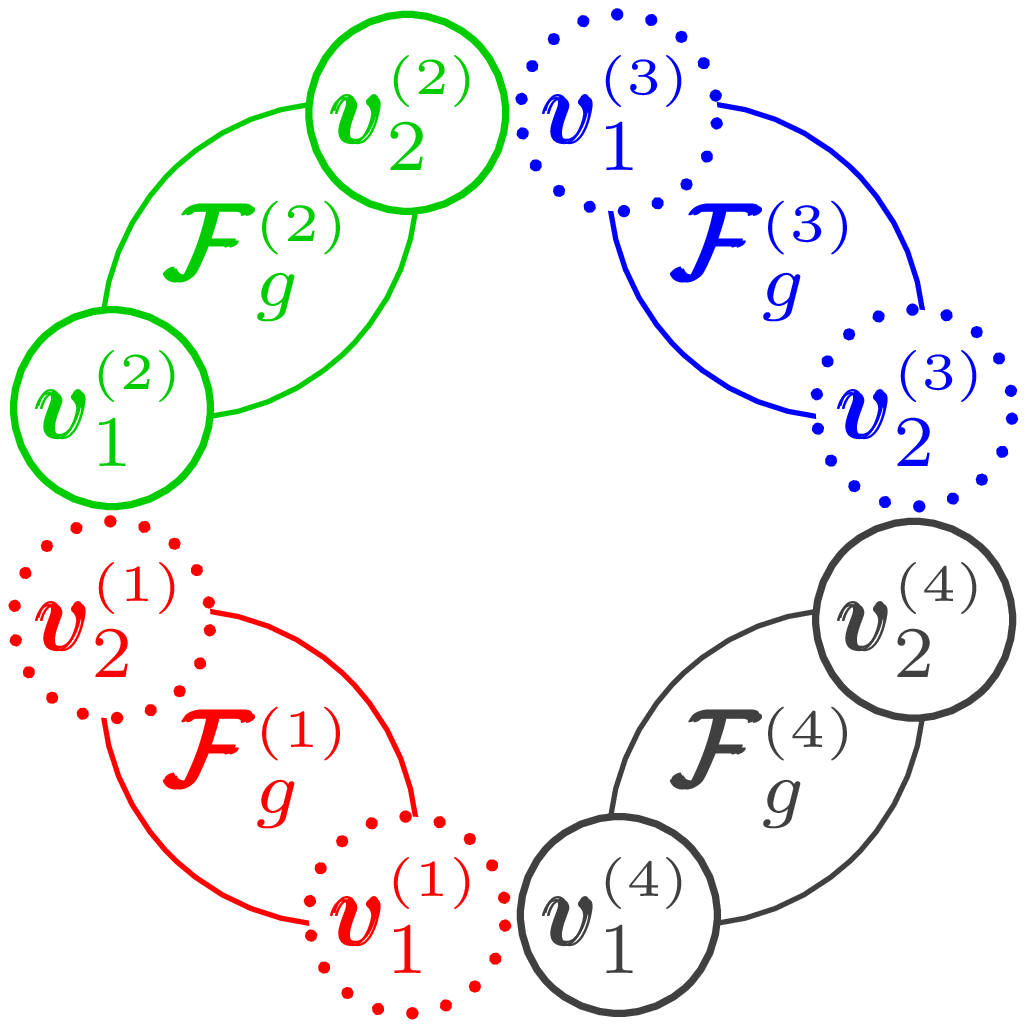}
	\end{subfigure}
	\caption{Configurations of maximum matchings for $\mathcal{F}_{g+1}$ with size $2a_g+2c_g$.}
\end{subfigure}
\begin{subfigure}[b]{0.49\textwidth}
	\begin{subfigure}[b]{0.4\textwidth}
		\includegraphics[width=\textwidth]{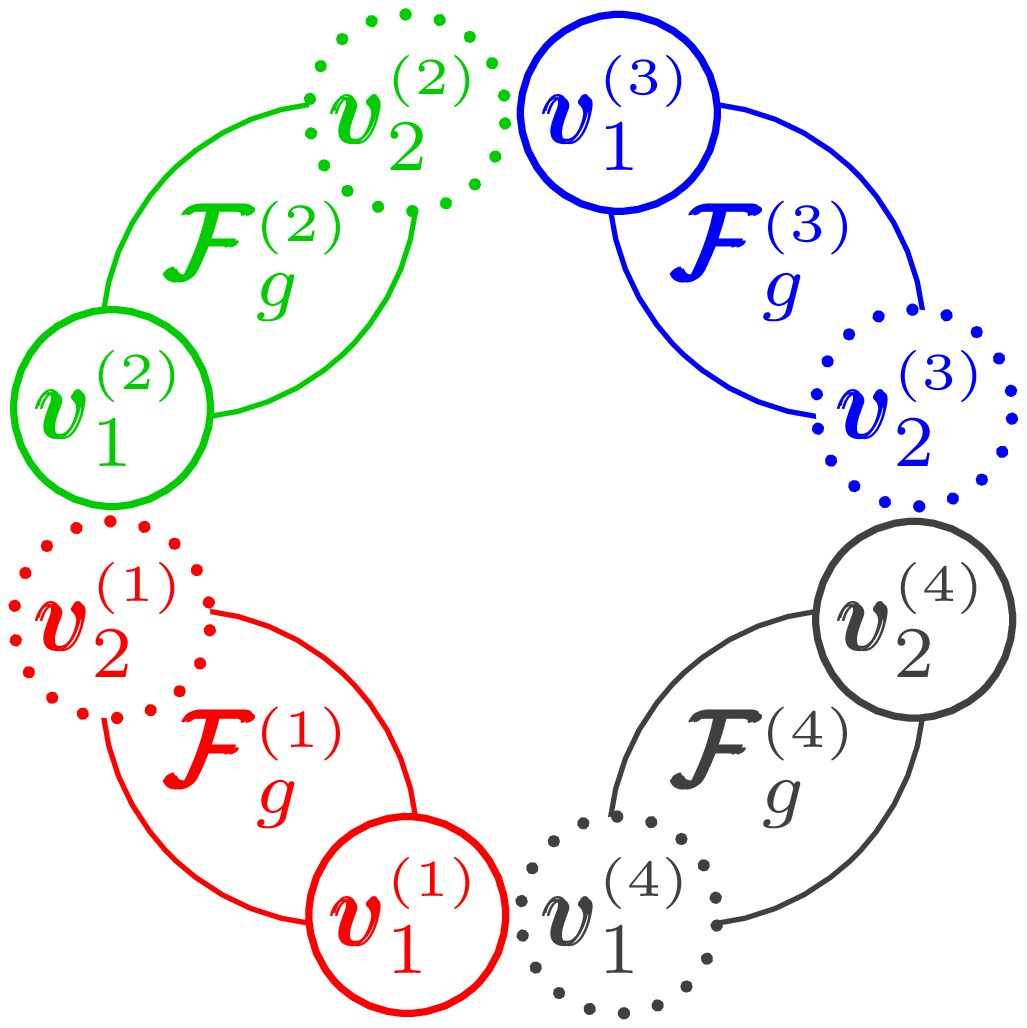}
	\end{subfigure}
	\begin{subfigure}[b]{0.4\textwidth}
		\includegraphics[width=\textwidth]{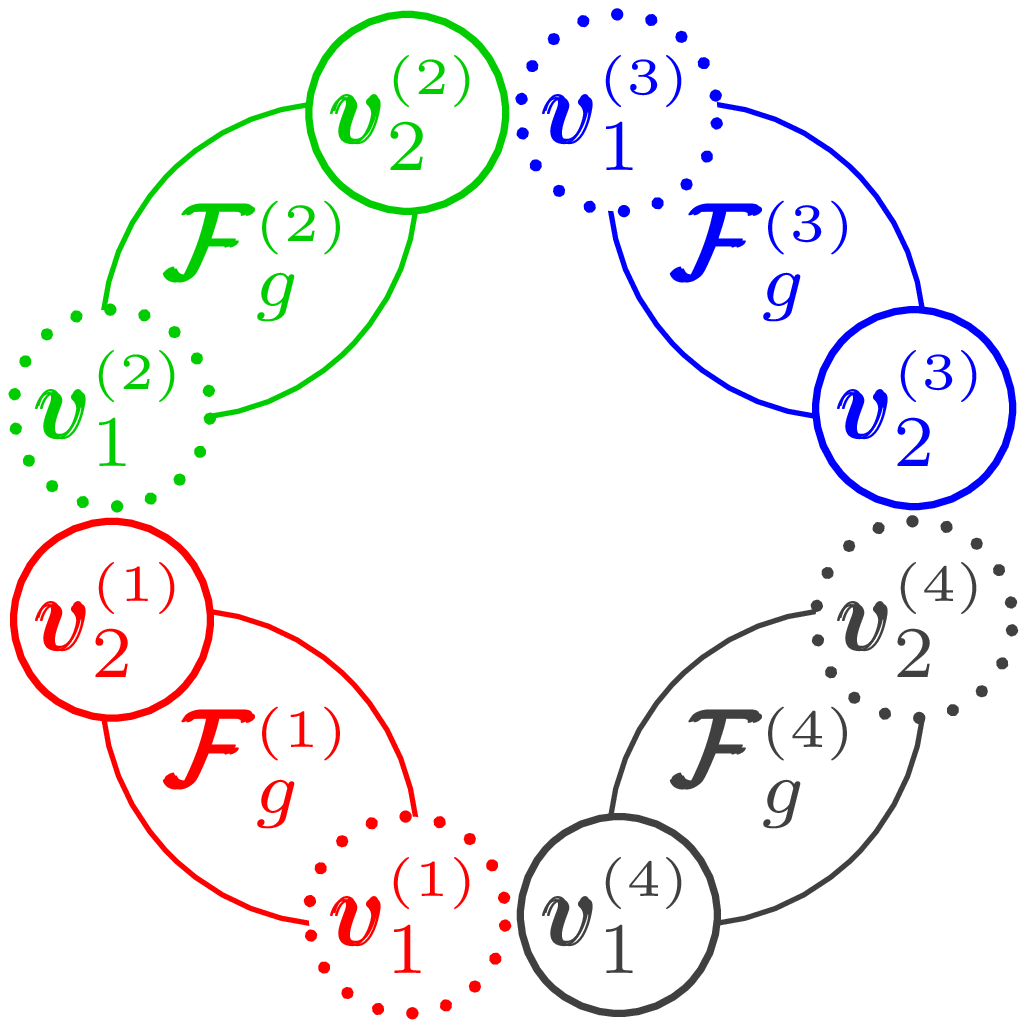}
	\end{subfigure}
	\caption{Configurations of maximum matchings for $\mathcal{F}_{g+1}$ with size $4b_g$.}
\end{subfigure}
\vfill
\null
\begin{subfigure}[b]{\textwidth}
    \centering
    \begin{subfigure}[b]{0.2\textwidth}
        \includegraphics[width=\textwidth]{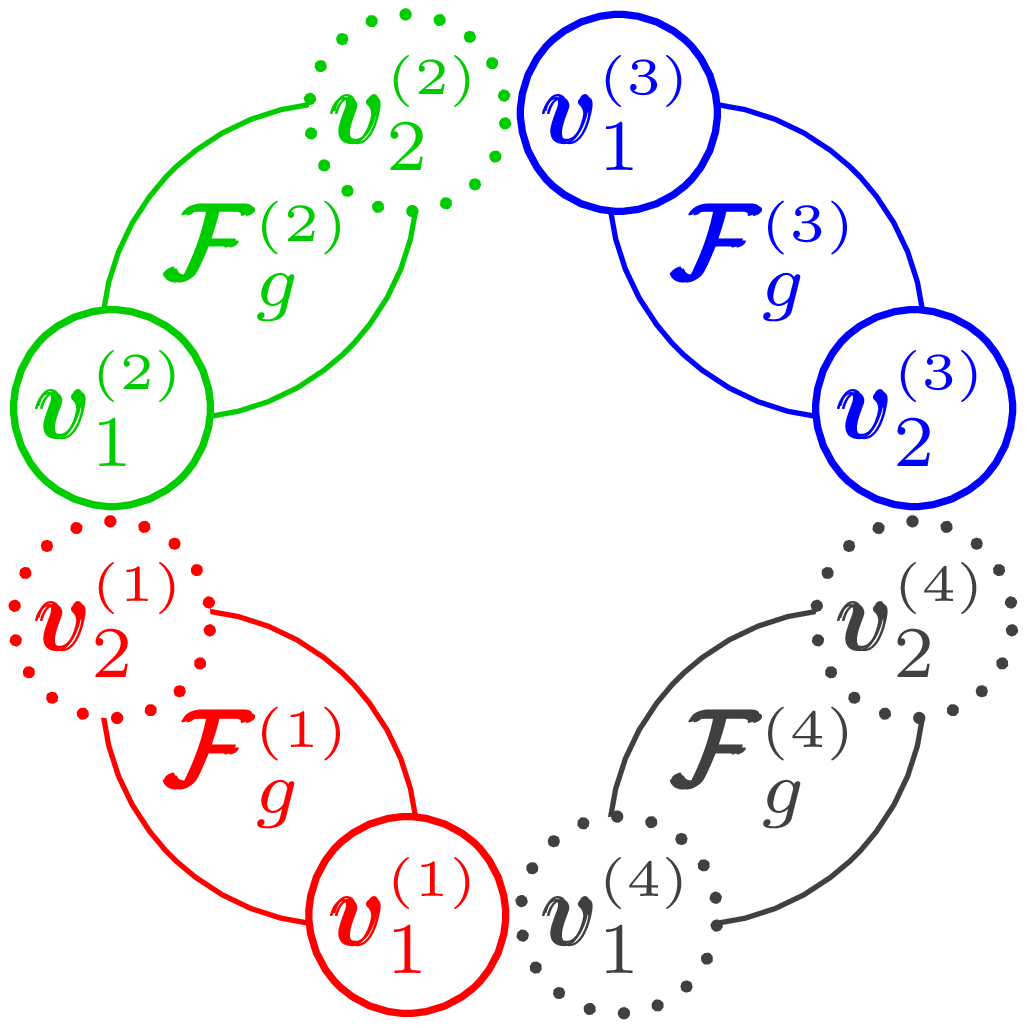}
    \end{subfigure}
    \begin{subfigure}[b]{0.2\textwidth}
        \includegraphics[width=\textwidth]{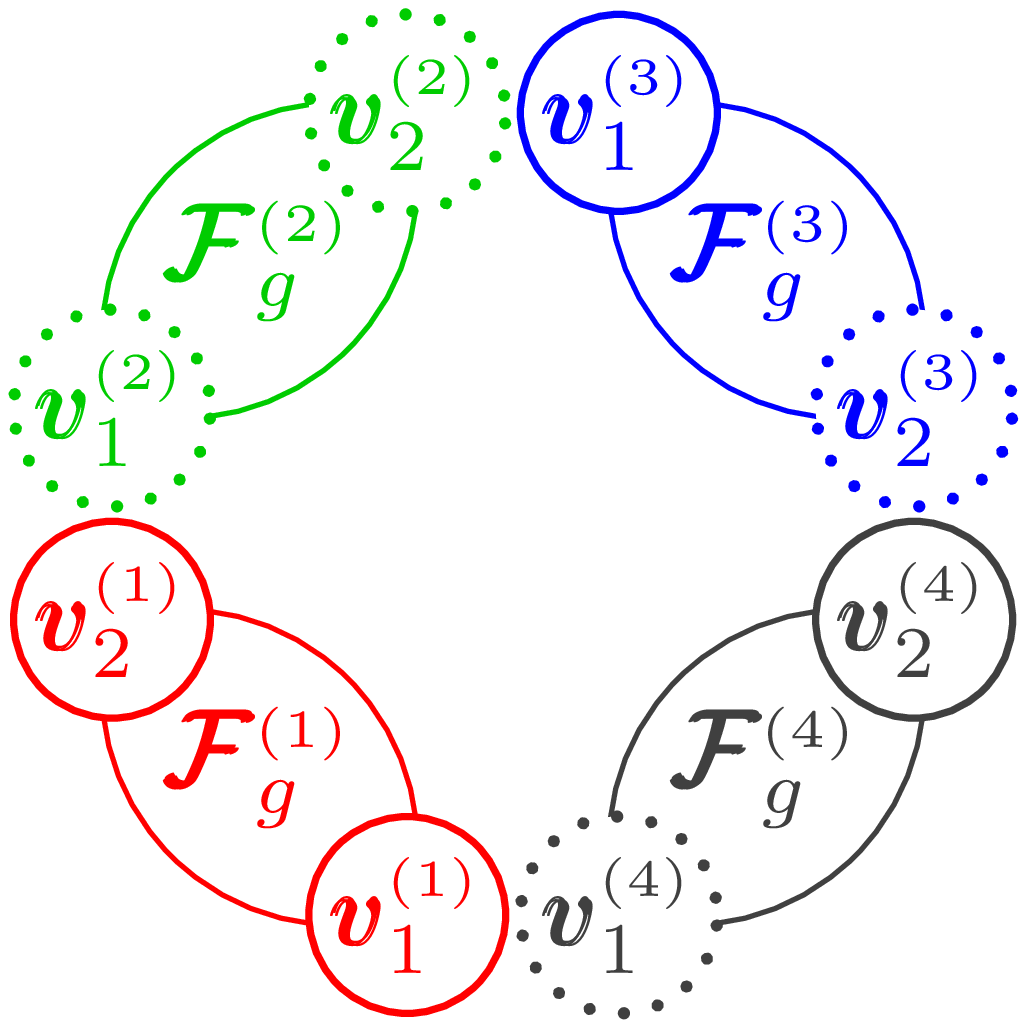}
    \end{subfigure}
    \begin{subfigure}[b]{0.2\textwidth}
        \includegraphics[width=\textwidth]{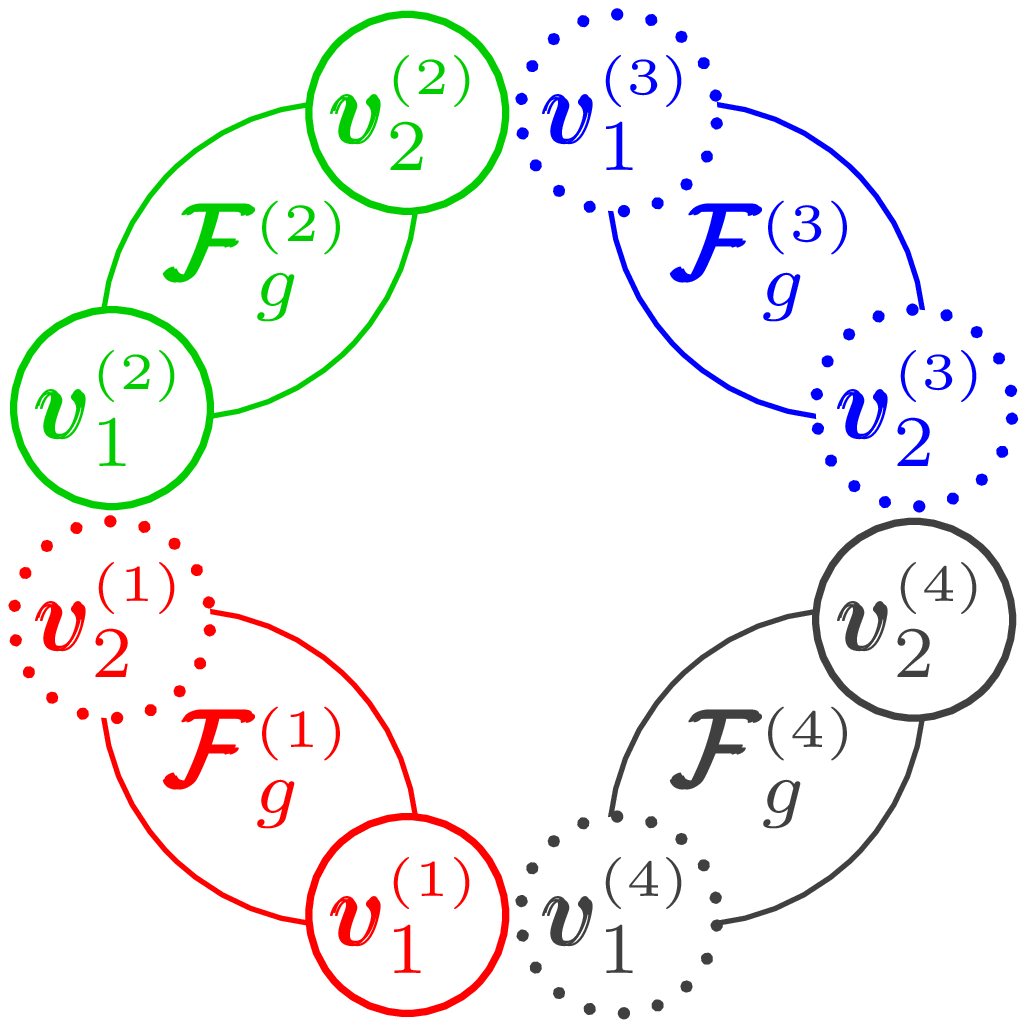}
    \end{subfigure}
    \begin{subfigure}[b]{0.2\textwidth}
        \includegraphics[width=\textwidth]{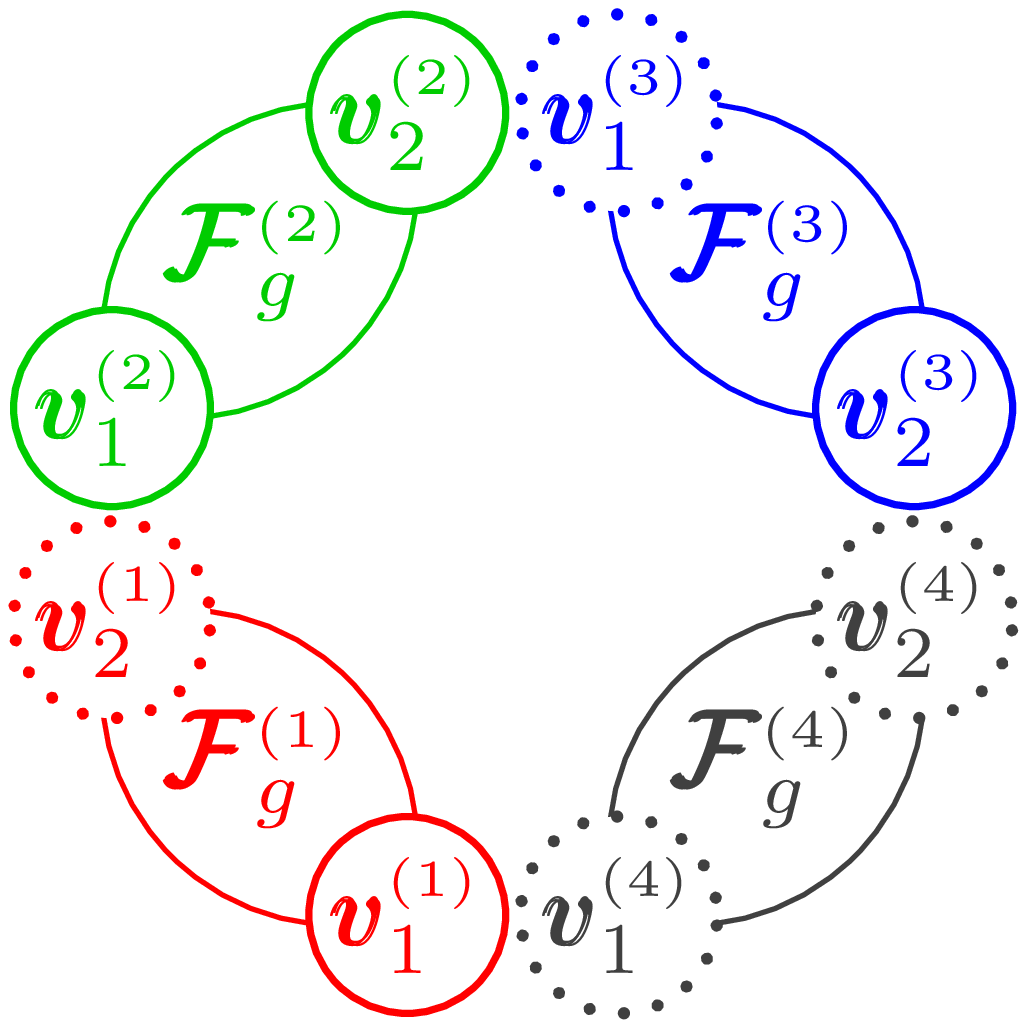}
    \end{subfigure}
    \vfill
    \null
    \begin{subfigure}[b]{0.2\textwidth}
        \includegraphics[width=\textwidth]{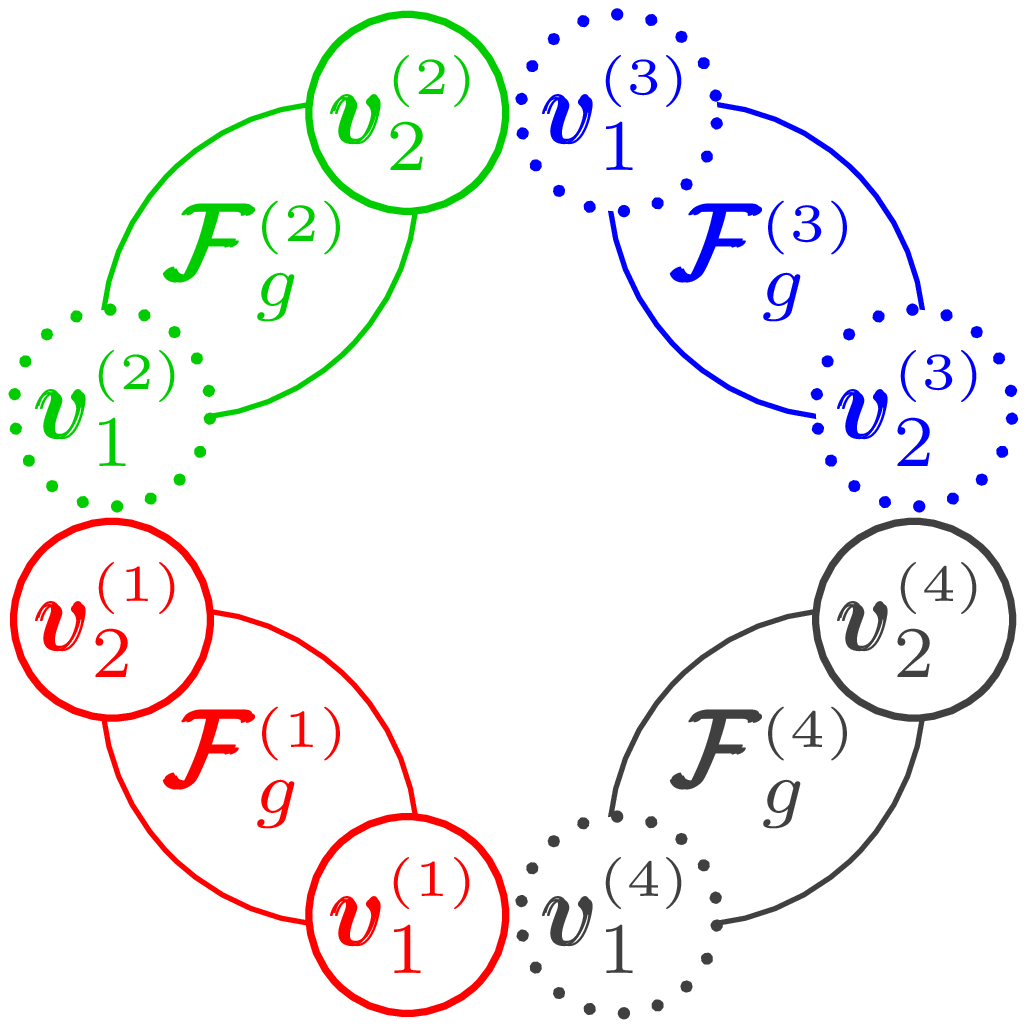}
    \end{subfigure}
    \begin{subfigure}[b]{0.2\textwidth}
        \includegraphics[width=\textwidth]{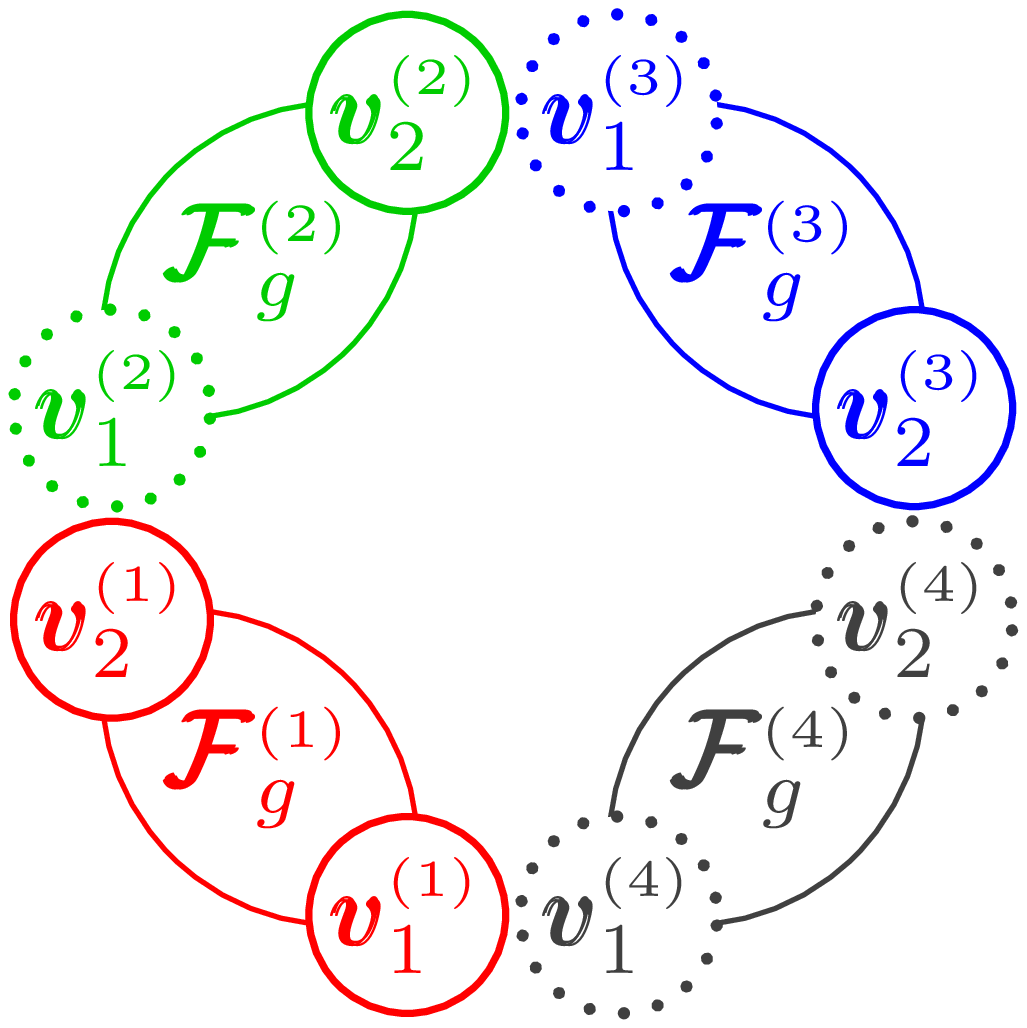}
    \end{subfigure}
    \begin{subfigure}[b]{0.2\textwidth}
        \includegraphics[width=\textwidth]{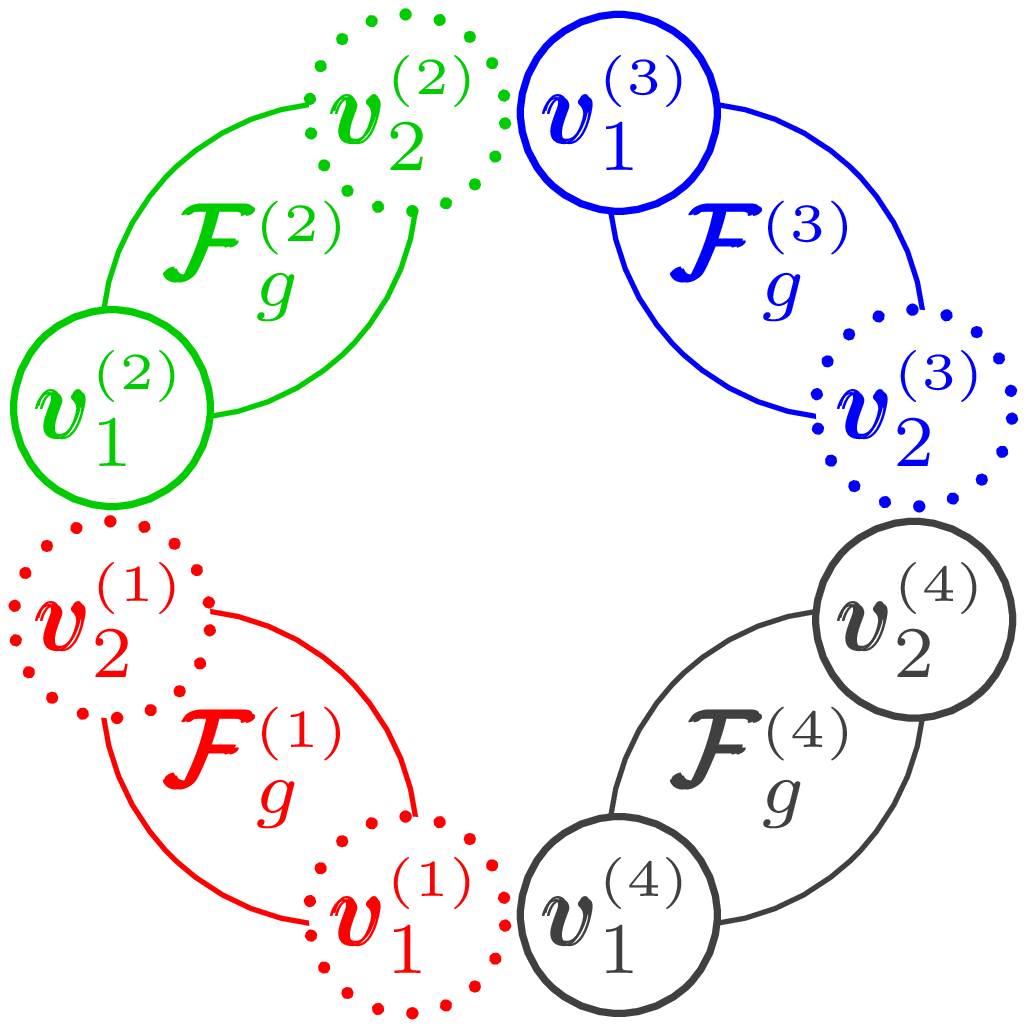}
    \end{subfigure}
    \begin{subfigure}[b]{0.2\textwidth}
        \includegraphics[width=\textwidth]{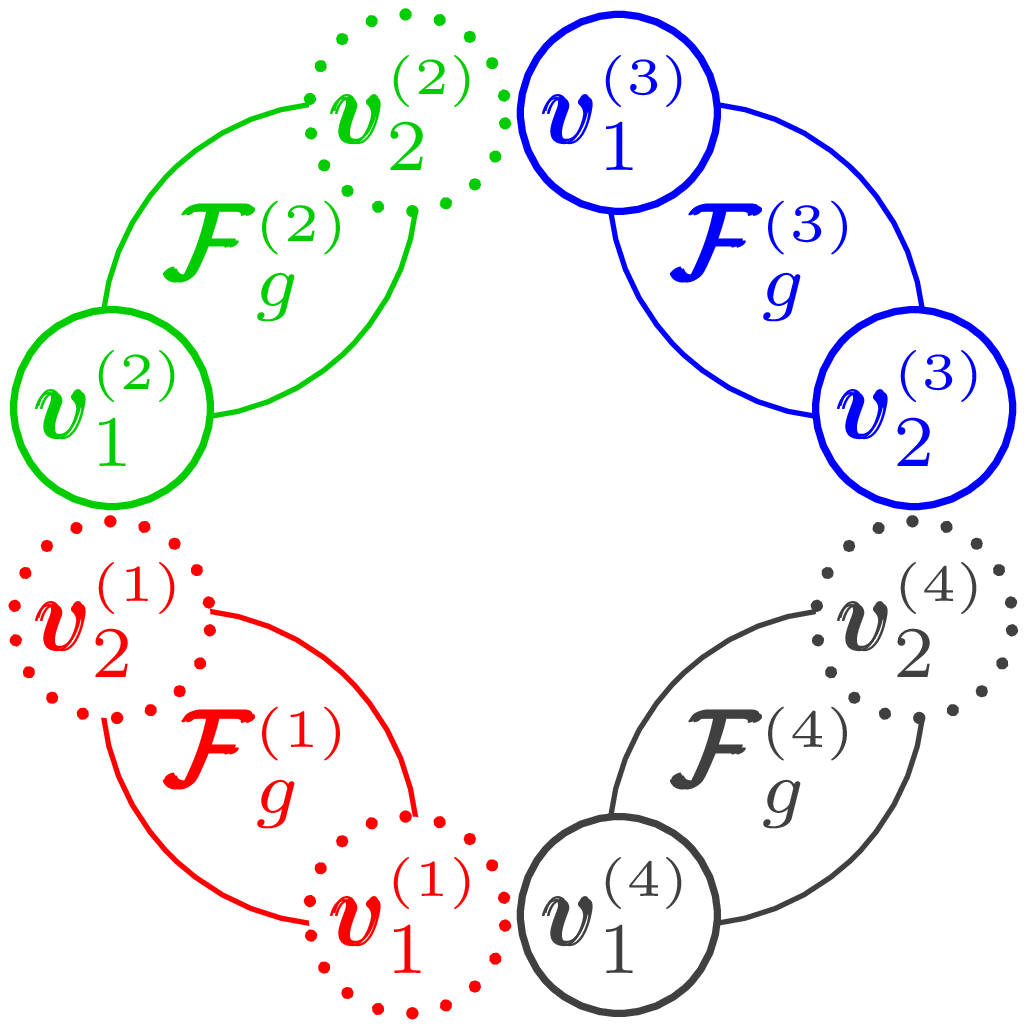}
    \end{subfigure}
    \vfill
    \null
    \begin{subfigure}[b]{0.2\textwidth}
        \includegraphics[width=\textwidth]{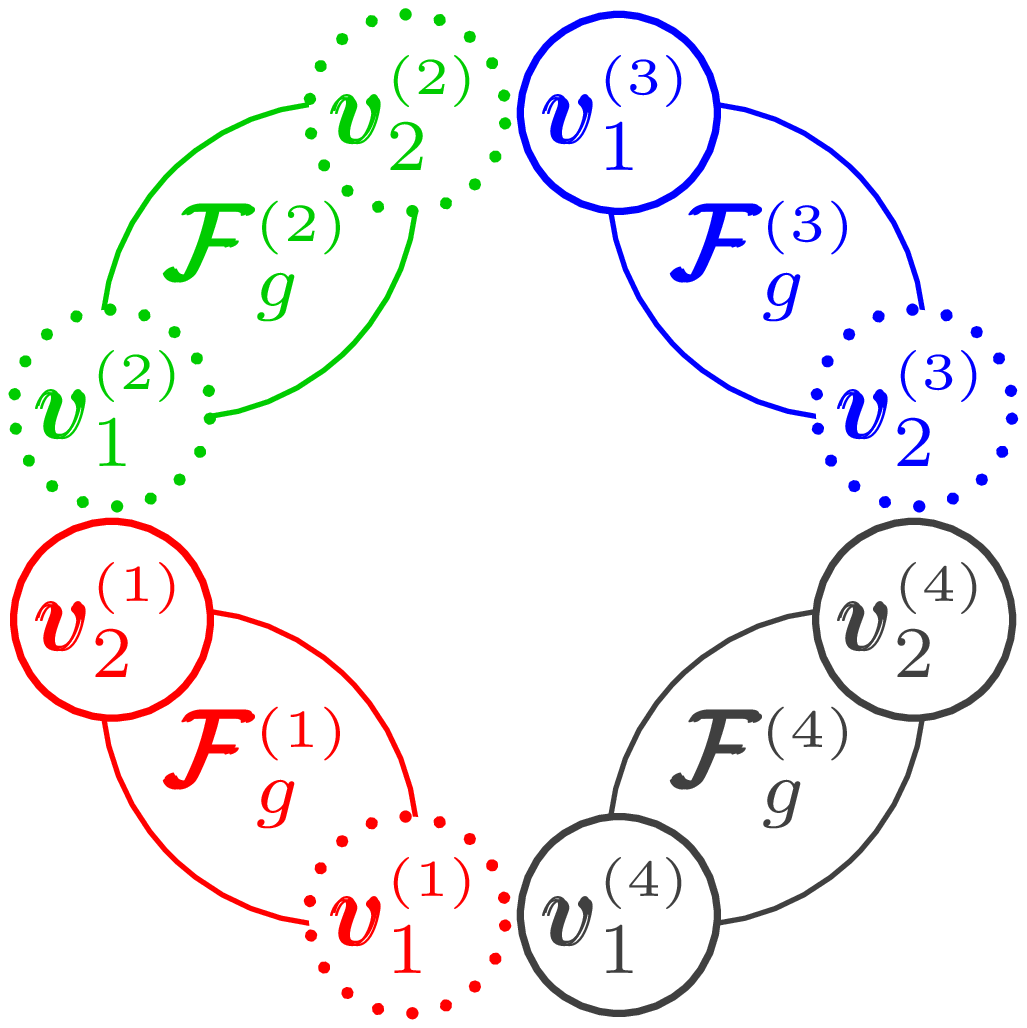}
    \end{subfigure}
    \begin{subfigure}[b]{0.2\textwidth}
        \includegraphics[width=\textwidth]{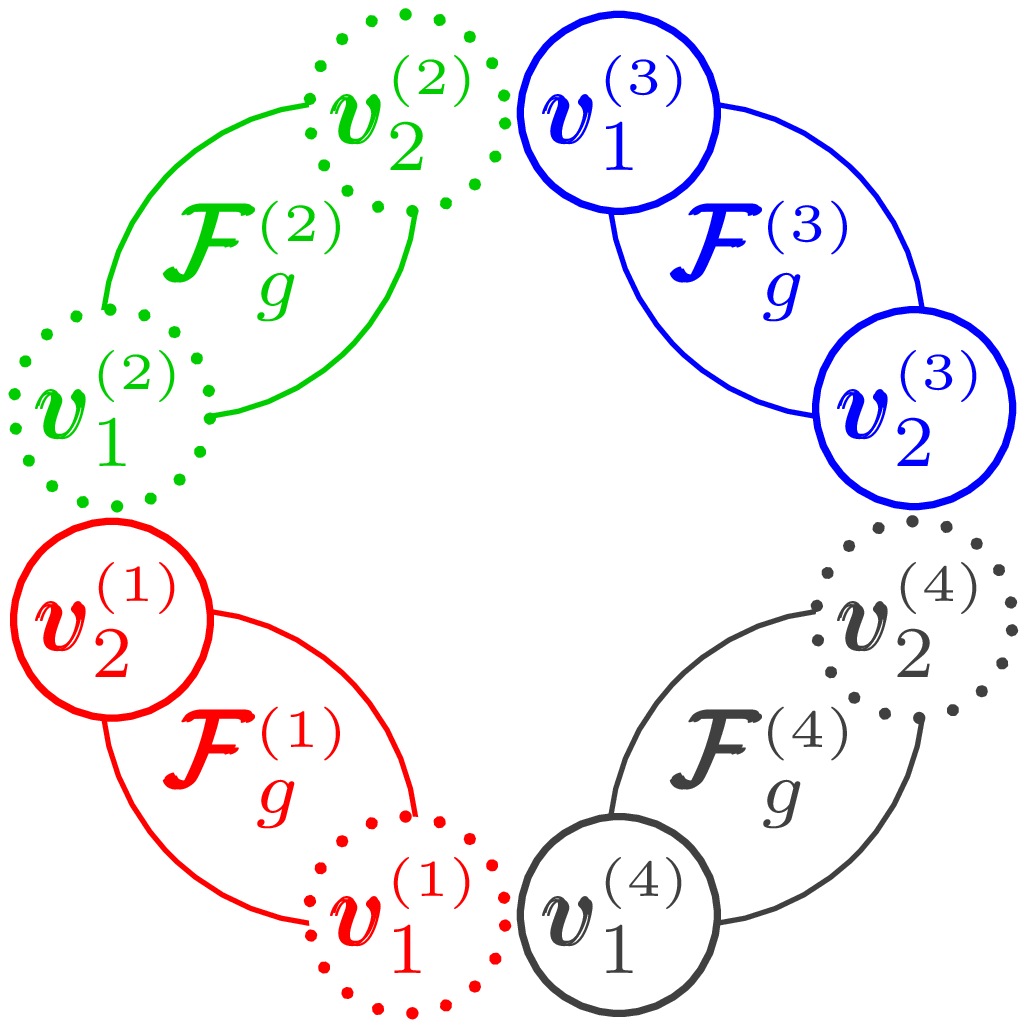}
    \end{subfigure}
    \begin{subfigure}[b]{0.2\textwidth}
        \includegraphics[width=\textwidth]{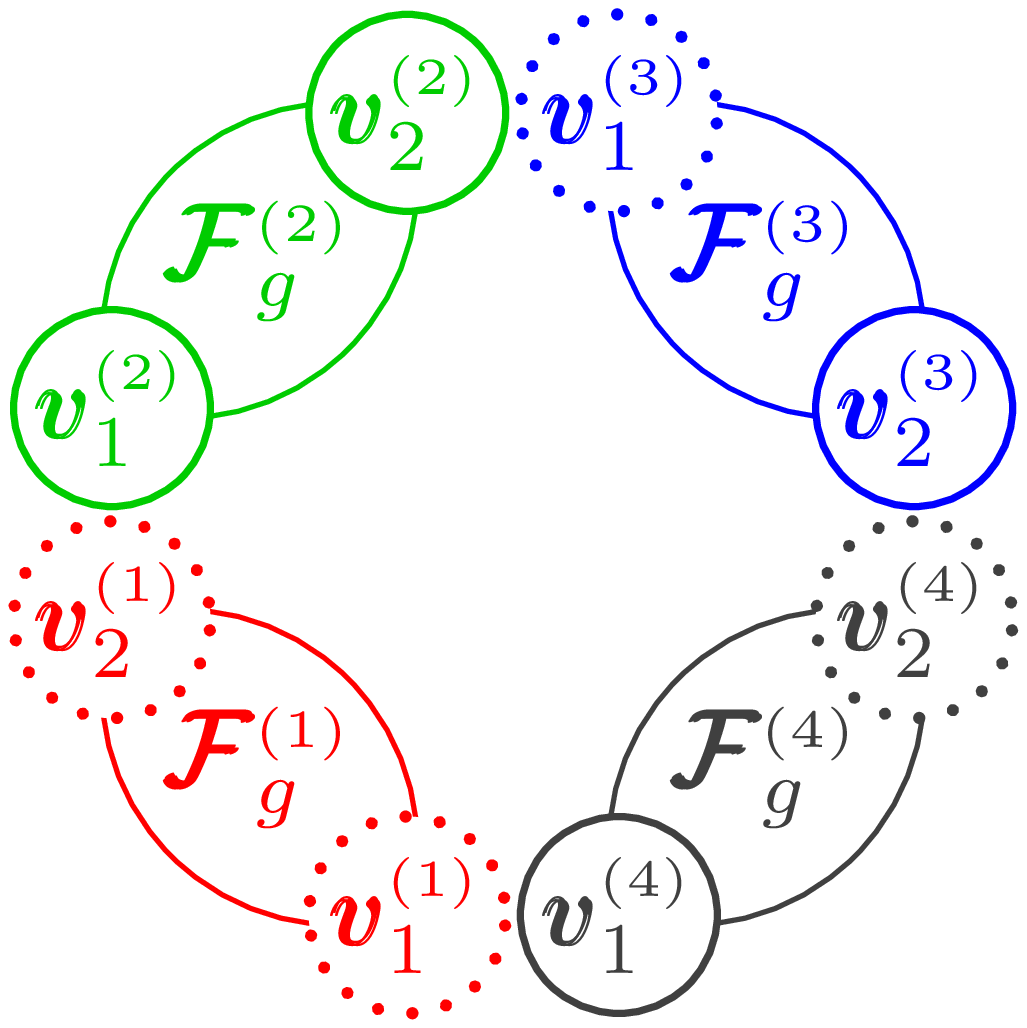}
    \end{subfigure}
    \begin{subfigure}[b]{0.2\textwidth}
        \includegraphics[width=\textwidth]{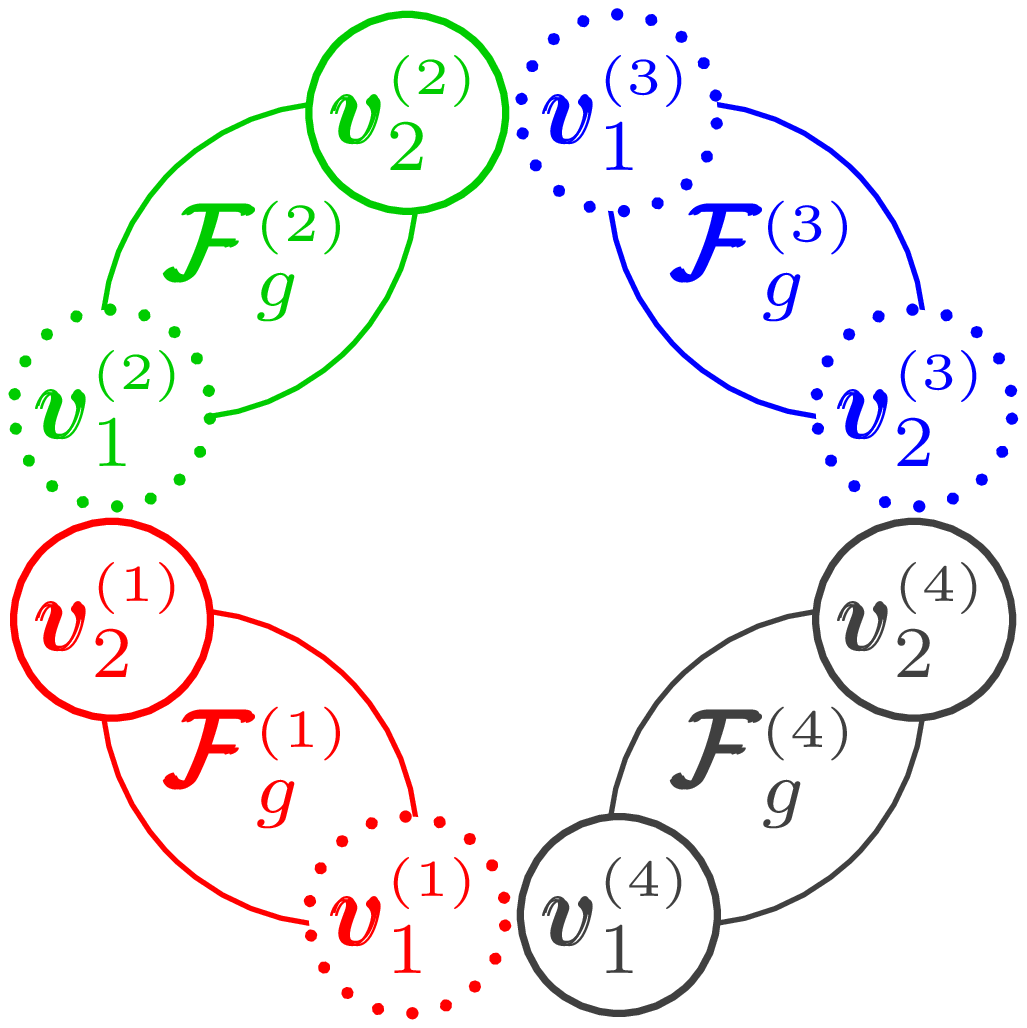}
    \end{subfigure}
    \caption{Configurations of maximum matchings for $\mathcal{F}_{g+1}$ with size \(a_g+2b_g+c_g\).}
\end{subfigure}
\caption{Illustration of all possible configurations of maximum matchings for $\mathcal{F}_{g+1}$.}
\label{cgeq}
\end{figure}

\subsection{Number of maximum matchings}

Let $\theta_g$ be the number of maximum matchings of $\mathcal{F}_g$. To determine $\theta_g$, we introduce two additional quantities.  Let $\phi_g$ be the number of maximum matchings of $\mathcal{F}_g\setminus \{v_1,v_2\}$, and let $\varphi_g$ be the number of maximum matchings of $\mathcal{F}_g\setminus \{v_1\}$, which is equal to the number of maximum matchings of $\mathcal{F}_g\setminus \{v_2\}$.

\begin{theo}\label{TheoNMMA}
For network $\mathcal{F}_g$, $g \geq 1$, the three quantities $\phi_g$, $\varphi_g$ and $\theta_g$ can be determined recursively according to the following relations:
\begin{eqnarray}\label{countingmm}
\phi_{g+1} &=& 4\phi_g^2\varphi_g^2, \notag \\
\varphi_{g+1} &=& 4\phi_g^2\varphi_g\theta_g + 4\phi_g\varphi_g^3, \notag \\
\theta_{g+1} &=& 2\phi_g^2\theta_g^2 + 2\varphi_g^4 + 12\phi_g\varphi_g^2\theta_g,\notag
\end{eqnarray}
with initial conditions $\phi_1 = 1$, $\varphi_1 = 2$
and $\theta_1 = 2$.
\end{theo}
\begin{proof}
Theorem~\ref{theomm} show that $a_g = \frac{4^g-4}{6}$, $b_g = \frac{4^g+2}{6}$ and $c_g = \frac{4^g+8}{6}$, which means
\begin{align*}
a_{g+1} &= 2a_g+2b_g, \\
b_{g+1} &= 2a_g+b_g+c_g = a_g+3b_g, \\
c_{g+1} &= 2a_g+2c_g = 4b_g = a_g+2b_g+c_g.
\end{align*}
Thus, Fig.~\ref{ageq},~\ref{bgeq} and~\ref{cgeq} actually provide  arrangements of all non-overlapping  maximum matchings for $\mathcal{F}_g\setminus \{v_1,v_2\}$, $\mathcal{F}_g\setminus \{v_1\}$, and $\mathcal{F}_g$, respectively.
From these three figures we can obtain directly the recursive relations for $\phi_g$, $\varphi_g$ and $\theta_g$. {According to Fig.~\ref{ageq}, we can see that, for the four configurations their contribution to  $\phi_{g+1}$  is identical.  By  using addition and multiplication principles, we obtain
\begin{align*}
 \phi_{g+1}= 4\phi_g^2\varphi_g^2.
\end{align*}
In a similar way, we can derive the recursive relations for $\varphi_{g+1}$ and $\theta_{g+1}$:
\begin{align}
\varphi_{g+1} &= 4\phi_g^2\varphi_g\theta_g + 4\phi_g\varphi_g^3, \notag \\
\theta_{g+1} &= 2\phi_g^2\theta_g^2 + 2\varphi_g^4 + 12\phi_g\varphi_g^2\theta_g,\notag
\end{align}
as stated by the theorem.}
\end{proof}


Theorem~\ref{TheoNMMA} shows that the number
of maximum matchings of network $\mathcal{F}_g$ can be calculated in $\mathcal{O}(\ln N_g)$ time.

\section{Pfaffian orientation and perfect matchings of a non-fractal scale-free network}

In the preceding section, we study the size and number of maximum matchings of  a fractal scale-free network $\mathcal{F}_g$, which has no perfect matching for all $g>1$. In this section, we study maximum matchings in the non-fractal counterpart $\mathcal{H}_g$ of $\mathcal{F}_g$, {which has an  infinite fractal dimension~\cite{ZhZhZoChGu09}.} Both $\mathcal{F}_g$ and $\mathcal{H}_g$ have the same degree sequence for all $g\geq 1$, but $\mathcal{H}_g$ has perfect matchings. Moreover, we determine the number of perfect matchings in $\mathcal{H}_g$ by applying the Pfaffian method, and show that $\mathcal{H}_g$ has the same entropy for perfect matchings as that associated with   an extended Sierpi\'nski graph~\cite{KlMo05}.

\subsection{Network construction and structural characteristics}

The non-fractal scale-free network is also built iteratively. 

\begin{definition}\label{Def:NF01}
Let $\mathcal{H}_g$ be the non-fractal scale-free network $\mathcal{H}_g=(\mathcal{V}(\mathcal{H}_g),\,\mathcal{E}(\mathcal{H}_g))$, $g \geq 1$,
after $g$ iterations, with vertex set $\mathcal{V}(\mathcal{H}_g)$ and edge set $\mathcal{E}(\mathcal{H}_g)$. Then, $\mathcal{H}_g$ is constructed in the following iterative way.

For $g=1$, $\mathcal{H}_1$ consists of a quadrangle.

For $g>1$, $\mathcal{H}_g$ is derived from $\mathcal{H}_{g-1}$ by replacing each edge in $\mathcal{H}_{g-1}$ with a quadrangle on rhs of the arrow in Fig.~\ref{Fig.1}.
\end{definition}

Figure~\ref{Fig.2} illustrates three networks for $g=1,2,3$.

\begin{figure}
\centering
\includegraphics[width=0.4\textwidth]{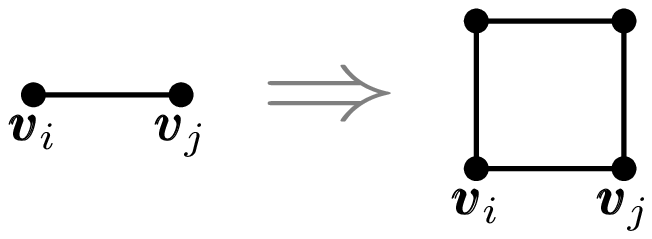}
\caption{Iterative construction method for the non-fractal scale-free network. Each edge  $(v_i,v_j)$ generates two vertices, which, together with $v_i$ and $v_j$, form a quadrangle on the rhs of the arrow.}
\label{Fig.1}
\end{figure}

\begin{figure}
\centering
\includegraphics[width=0.6\textwidth]{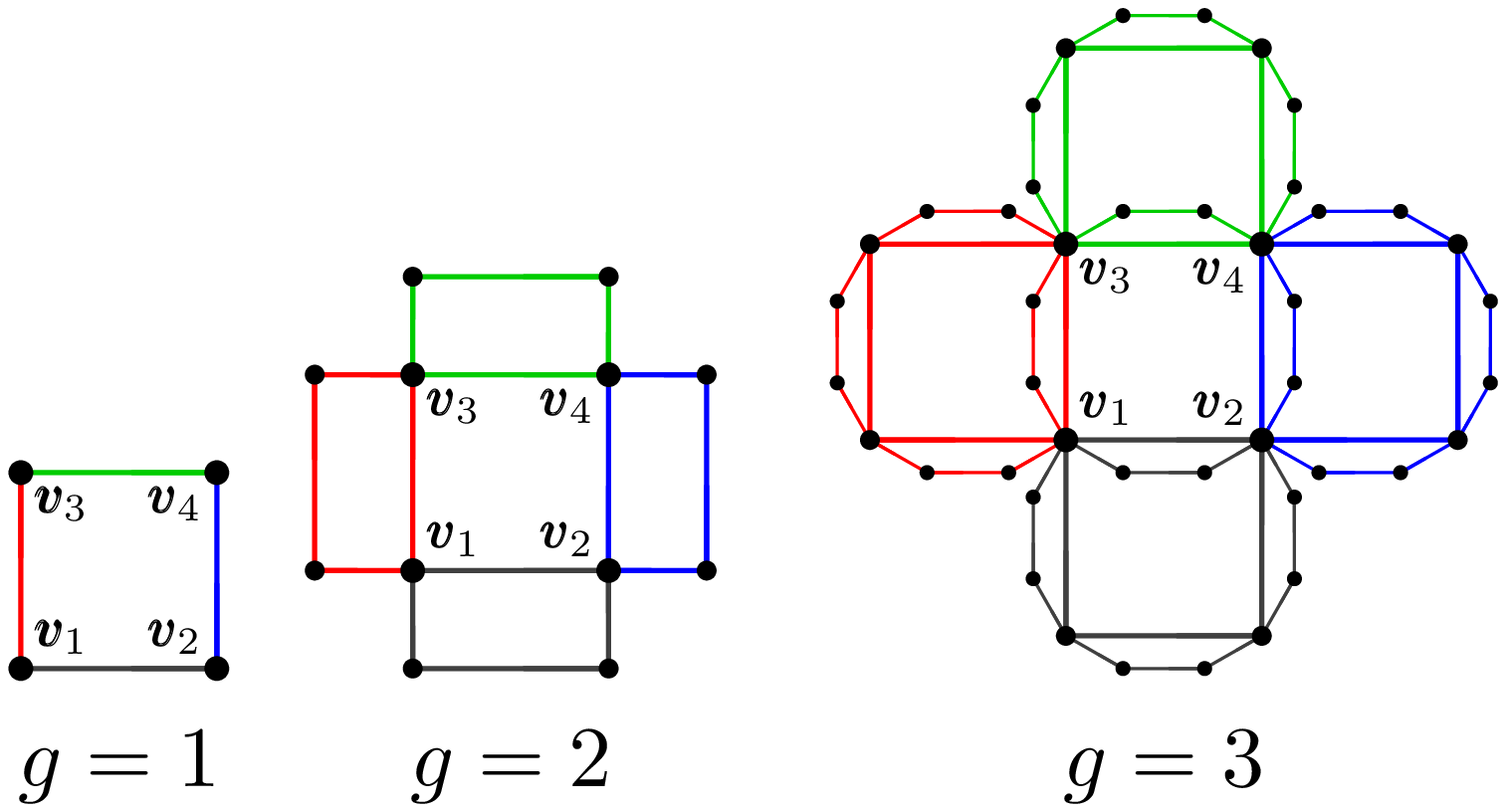}
\caption{The non-fractal scale-free network $\mathcal{H}_{1}$, $\mathcal{H}_{2}$, and $\mathcal{H}_{3}$.}
\label{Fig.2}
\end{figure}

The non-fractal scale-free network is also self-similar, which can also be generated in an alternative approach~\cite{ZhZhZoChGu09}. Similar to its fractal counterpart $\mathcal{F}_g$, in $\mathcal{H}_g$, $g \geq 1$, the initial four vertices created at $g =1$ have the largest degree, which are call hub vertices. We label the four hub vertices in $\mathcal{H}_1$ by $v_1$, $v_2$, $v_3$, and $v_4$: one pair of diagonal vertices are label as $v_1$ and $v_4$, and the other pair of vertices are labeled as $v_2$ and $v_3$, see Fig.~\ref{Fig.2}. Then, the non-fractal scale-free network can be created alternatively as shown in Fig.~\ref{Fig.3}.

\begin{figure}
\centering
\includegraphics[width=0.7\textwidth]{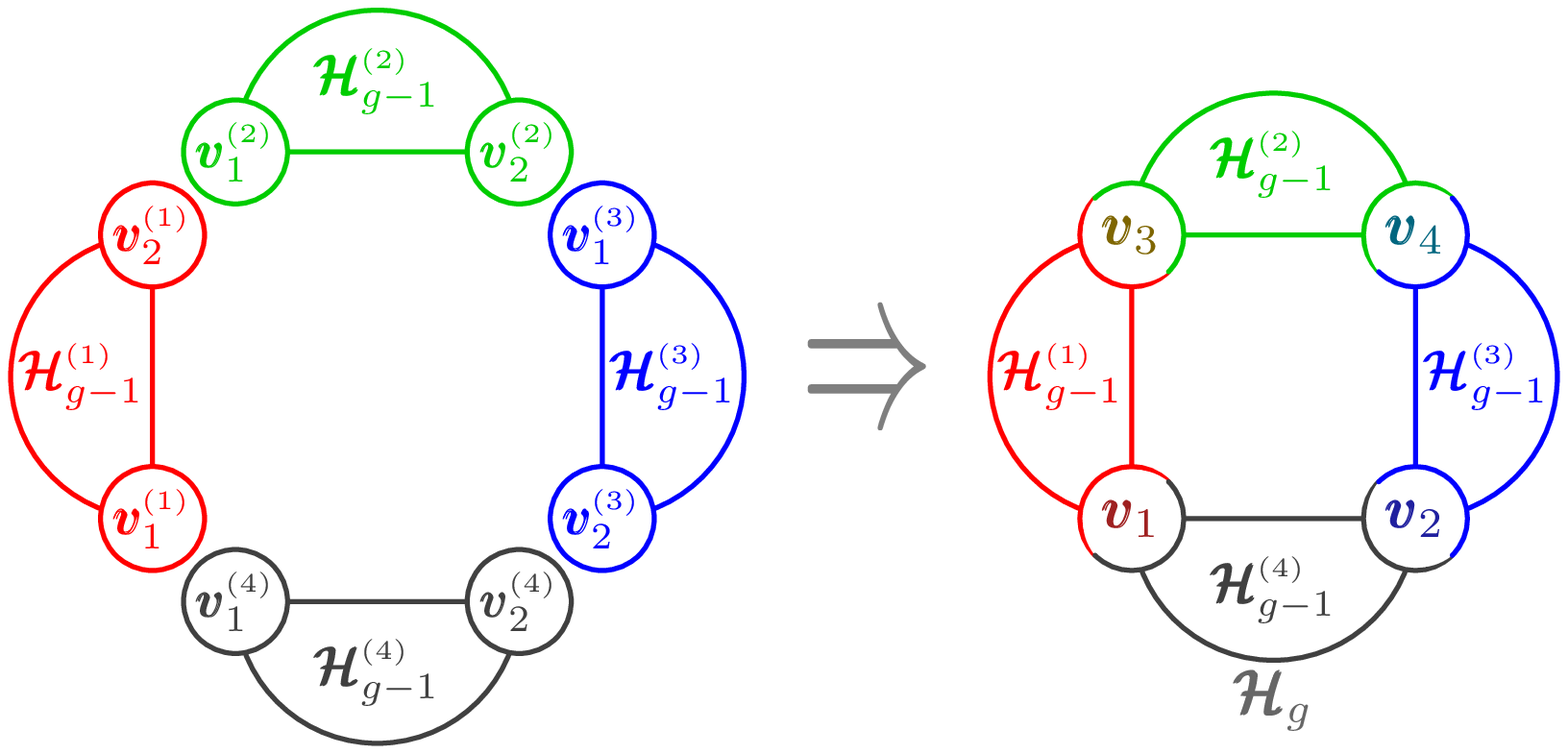}
\caption{Another construction approach for the non-fractal scale-free  network.} 
\label{Fig.3}
\end{figure}

\begin{definition}\label{Def:NF02}
Given the network $\mathcal{H}_{g-1}=(\mathcal{V}(\mathcal{H}_{g-1}),\,\mathcal{E}(\mathcal{H}_{g-1}))$, $g > 1$, $\mathcal{H}_g=(\mathcal{V}(\mathcal{H}_g),\,\mathcal{E}(\mathcal{H}_g))$ is obtained by executing the following two operations:

(i)  Amalgamating four copies of $\mathcal{H}_{g-1}$, denoted by $\mathcal{H}_{g-1}^{(i)}$, $i=1,2,3,4$, the four hub vertices of which are denoted by $v_k^{(i)}$, $k=1,2,3,4$, with $v_k^{(i)}$ in $\mathcal{H}_{g-1}^{(i)}$ corresponding to $v_k$ in $\mathcal{H}_{g-1}$.

(ii)  Identifying $v_1^{(1)}$ and $v_1^{(4)}$ (or $v_2^{(2)}$ and $v_1^{(3)}$, $v_2^{(1)}$ and $v_1^{(2)}$, $v_2^{(3)}$ and $v_2^{(4)}$) as the hub vertex $v_1$ (or $v_4$,  $v_3$,  $v_2$) in $\mathcal{H}_g$.
\end{definition}


{In the sequel, we will use the above  notations for $\mathcal{F}_g$ to represent the same  quantities corresponding to those of $\mathcal{H}_g$ in the case without confusion.}

In $\mathcal{H}_g$, the number of vertices is $N_g=\frac{2}{3}\left(4^{g}+2 \right)$, the number of edges is $E_g=4^g$. According to the first construction method, the number of vertices created at iteration $g_i$, $ g_i >1$, is $L_v(g_i)=2\times 4^{ g_i -1}$, the degree of a vertex $i$ created at iteration $g_i$, $ g_i \geq 1$, is $d_i(g)= 2^{g-g_i+1}$, all possible degree of vertices is $2^{ g-g_i+1}$, $ 1 \leq g_i \leq g$, and the number of vertices with degree is $L_v(g_i)$.

As shown in~\cite{ZhZhZoChGu09}, for all $g\geq 1$, $\mathcal{H}_g$ has the same degree sequence as that of $\mathcal{F}_g$. Thus, $\mathcal{H}_g$ is scale-free with the power exponent $\gamma=3$, identical to that of $\mathcal{F}_g$.

In spite of the resemblance of degree sequence between $\mathcal{H}_g$ and $\mathcal{F}_g$, there are obvious difference between them. For example,
$\mathcal{H}_g$ is non-fractal since its fractal dimension is infinite~\cite{ZhZhZoChGu09}. Another example is that network
$\mathcal{H}_g$ is typically small-world.
\begin{proposition}\cite{ZhZhZoChGu09}
The average distance of $\mathcal{H}_g$ is
$$
\mu(\mathcal{H}_g) = \frac{2}{3} \times \frac{8+16\times 4^g+3\times 16^g + 6g\times 16^g}{4\times 16^g + 10\times 4^g + 4}.
$$
\end{proposition}
For infinite $g$, $\mu(\mathcal{H}_g)$ approximates $g$, and thus increases logarithmically with number of vertices, implying that network $\mathcal{H}_g$ exhibits a small-world behavior.

Interestingly, network $\mathcal{H}_g$ is absolutely uncorrelated.
\begin{proposition}
In network $\mathcal{H}_g$, $g\geq 1$, the average degree of all the neighboring vertices for vertices with degree $d$ is
$$k_{\rm nn} (d) = g+1,$$
independent of $d$.
\label{corrHg0}
\end{proposition}
\begin{proof}
Note that in network $\mathcal{H}_g$, the possible degree is $d_i(g)=2^{g-g_i+1}$ with $ 1 \leq g_i \leq g$.  {For any $d_i(g)=2^{g-g_i+1}$, let $k_{\rm nn} (d_i(g))$ be  the average degree of the neighboring vertices for  all the $L_v(g_i)$ vertices with degree $d_i(g)$.  Then, $k_{\rm nn} (d_i(g))$ is equal to the ratio of  the total degree of all neighbors of the $L_v(g_i)$ vertices having degree $d_i(g)$ to the total degree of these $L_v(g_i)$ vertices, given by}
\begin{align}\label{knn01}
k_{\rm nn} (d_i(g)) =\frac{1}{L_v(g_i)d_i(g)}&\times \Bigg(\sum\limits_{g'_{i}=1}^{g_{i}-1} {d(g'_{i},g)L_v (g'_{i})d(g'_{i},g_i-1)}+ \nonumber \\
&\quad \sum\limits_{g'_{i}=g_{i} + 1}^{g} {d(g'_{i}, g)L_v (g_{i})d(g_{i},g'_i-1)}\Bigg )+1,
\end{align}
where $d(x,y)$, $x \leq y$, represents the degree of a vertex in network $\mathcal{H}_y$, which was generated at iteration $x$. In Eq.~\eqref{knn01}, the first sum on the rhs  accounts for the links made to vertices with larger degree (i.e. $1\leq g'_{i} <g_{i}$) when the vertices was generated at iteration $g_{i}$. The second sum
explains the links made to the current smallest degree vertices
at each iteration $g'_{i}> g_{i}$. The last term 1 describes the link
connected to the simultaneously emerging vertex. After simple
algebraic manipulations, we have exactly $k_{\rm nn} (d_i(g)) = g+1$,
which does not depend on $d_i(g)$.
\end{proof}


The absence of degree correlations in network $\mathcal{H}_g$ can also be seen from its Pearson correlation coefficient.
\begin{proposition}
The Pearson correlation coefficient of network $\mathcal{H}_g$, $g\geq 1$, is
$r(\mathcal{H}_g)=0$.
\label{corrHg}
\end{proposition}
\begin{proof}
Using a similar proof process for Proposition \ref{corrFt}, we can determine related summations over all $E_g$ edges in $\mathcal{H}_g$:
\begin{equation*}
\sum \limits_{m=1}^{E_g} j_m k_m =\frac{1}{2} \sum \limits_{g_i=1}^{g} L_v(g_i) d^2_i(g) k_{\rm nn}(d_i(g))
 =  4^{g}(g+1)^2,
\end{equation*}
and
\begin{equation*}
\sum \limits_{m=1}^{E_g} (j_m  + k_m)  = \sum \limits_{i=1}^{N_g}d_i^2(g)
 =  \sum \limits_{g_i=1}^{g} L_v(g_i) d^2_i(g)
 =  (g+1)2^{2g+1}.
\end{equation*}
Then, according to Eq.~\eqref{rNe}, the numerator of $r(\mathcal{H}_g)$ is
\begin{equation*}
E_g\sum \limits_{m=1}^{E_g} j_m k_m - \left[\sum \limits_{m=1}^{E_g} \frac{1}{2} (j_m+ k_m)\right]^2= 4^g 4^{g} (g+1)^2 - \left[ 2^{2g}(g+1) \right]^2 =  0,
\end{equation*}
which leads to $r(\mathcal{H}_g)=0$.
\end{proof}



\subsection{Pfaffian orientation}

We here define an orientation of network $\mathcal{H}_g$, and then prove that the orientation is a Pfaffian orientation of $\mathcal{H}_g$, by using its self-similar structure.

\begin{definition}\label{Def:LW02}
The orientation $\mathcal{H}_g^e$ of network $\mathcal{H}_g$, $g \geq 1$, is defined as follows:

For $g=1$, the orientations of four edges $(v_1,v_2)$, $(v_1,v_3)$, $(v_2,v_4)$, and $(v_3,v_4)$ in $\mathcal{H}_1$ are, respectively, from $v_1$ to $v_2$, from $v_1$ to $v_3$, $v_4$ to $v_2$, and $v_3$ to $v_4$.

For $g>1$, $\mathcal{H}_g^e$ is obtained from $\mathcal{H}_{g-1}^e$. Note that $\mathcal{H}_g$ includes  four copies of $\mathcal{H}_{g-1}$, denoted by $\mathcal{H}_{g-1}^{(i)}$, $1\leq i \leq 4$. The orientations of  $\mathcal{H}_{g-1}^{(i)}$ are represented by $\mathcal{H}_{g-1}^{e\,(i)}$, each of which is a replica of $\mathcal{H}_{g-1}^e$.
\end{definition}

Fig.~\ref{Fig.4} shows the orientations of network $\mathcal{H}_g$ for $g=1,2,3$.

 \begin{figure}
 \centering
 \includegraphics[width=0.6\textwidth]{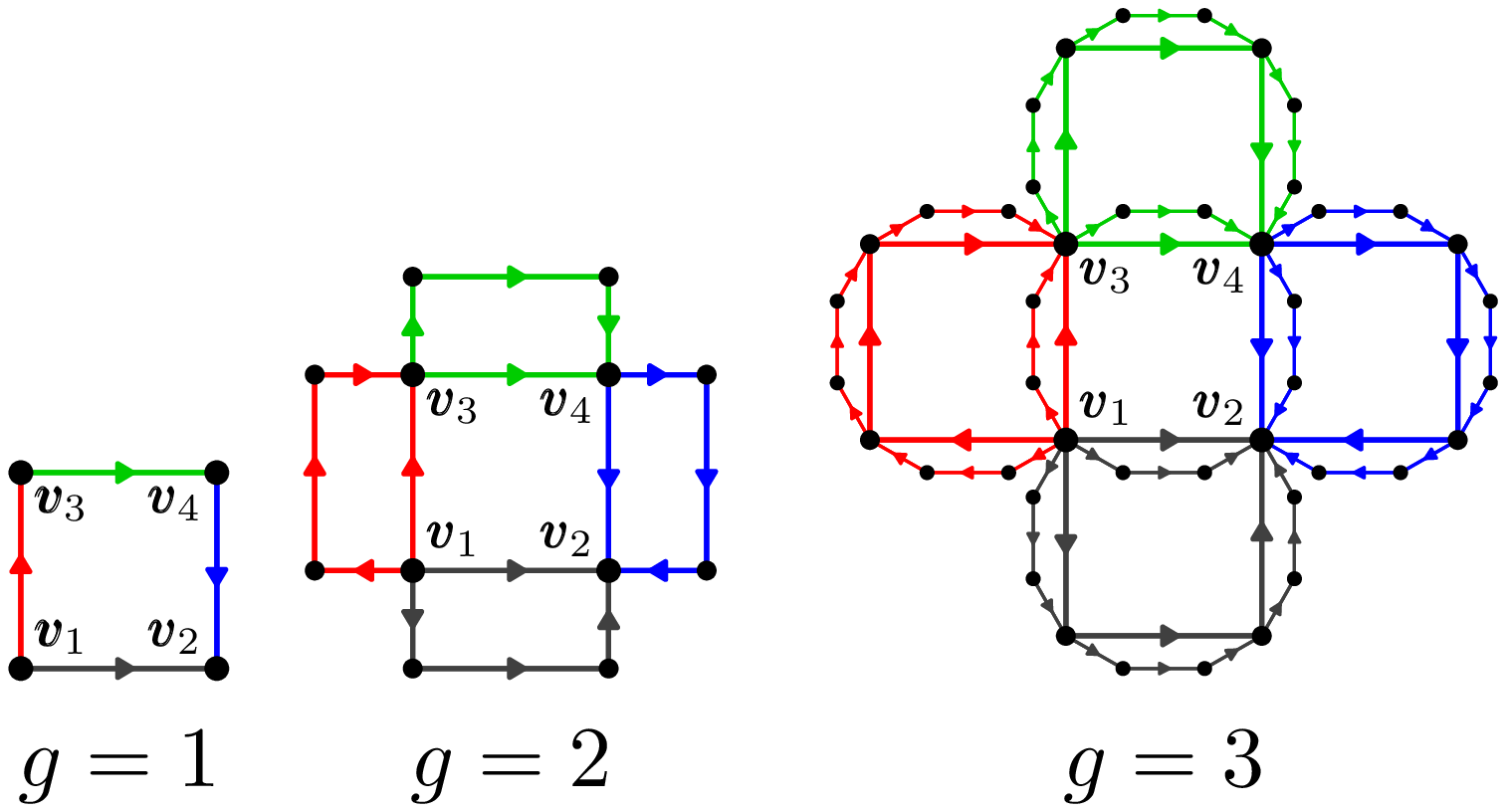}
 \caption{The orientation $\mathcal{H}_g^e$ of $\mathcal{H}_g$ for $g=1,2,3$.}
 \label{Fig.4}
 \end{figure}

Although no polynomial algorithm for checking whether a given orientation is Pfaffian or not is known~\cite{Lin200916}, for the network $\mathcal{H}_g$, we can prove that $\mathcal{H}_g^e$ is Pfaffian.

\begin{theo}\label{Lem:LW04}
For all $g\geq 1$, the orientation $\mathcal{H}_g^e$ is a Pfaffian orientation of network $\mathcal{H}_g$.
\end{theo}

{In order to prove Theorem~\ref{Lem:LW04}, we first prove that for any $g\geq 1$,  there exist perfect matchings for $\mathcal{H}_g$.  When $g=1$,  $\mathcal{H}_1$  is a quadrangle, which has two perfect matchings. Suppose that $\mathcal{H}_{g-1}$ ($g\geq 2$) has perfect matchings. If we keep the matching configurations  for vertices in  $\mathcal{H}_{g-1}$, we can cover those new vertices generated at iteration $g$ as follows. According to the first network construction approach, for any two vertices generated by an old  edge in   $\mathcal{H}_{g-1}$,  we cover this vertex pair by the new edge connecting them in $\mathcal{H}_g$. }

We further introduce some notations and give some auxiliary results. For arbitrary sequential hub vertices $u_1, u_2, \cdots, u_m$, let $S_{\mathcal{H}_g}(u_1, u_2, \cdots, u_m)$  present the set of paths (if $u_1\neq u_m$) or cycles (if $u_1=u_m$) of $\mathcal{H}_g$, where each path or cycle takes the form $u_1-\cdots-u_2-\cdots-u_m$,  exclusive of other hub vertices.  {Obviously, in $\mathcal{H}_g$ there exist nice paths starting from vertex $v_1$ to $v_2$. For example, the directed edge from $v_1$ to $v_2$ is a nice path.}
\begin{lemma}\label{Lem:LW03}
For $g\geq 1$, if $P$ is a nice path of $\mathcal{H}_g$ starting from vertex $v_1$ to $v_2$, then $P$ is oddly oriented relative to $\mathcal{H}_g^e$.
\end{lemma}
\begin{proof}
By induction. For $g=1$, it is obvious that the base case holds.

For $g>1$, suppose that the statement is true for $\mathcal{H}_{g-1}$. By construction, for any nice $P$ of $\mathcal{H}_g$ from vertex $v_1$ to $v_2$, it belongs to either of the two sets: $S_{\mathcal{H}_g}(v_1,v_2)$ and $S_{\mathcal{H}_g}(v_1,v_3,v_4,v_2)$.

For the first case $P \in S_{\mathcal{H}_g}(v_1,v_2)$, $P$ is evidently a nice path of $\mathcal{H}_{g-1}^{(4)}$ starting from vertex $v_1^{(4)}$ to vertex $v_2^{(4)}$. By induction hypothesis, $P$ is oddly oriented.

For the second case $P \in S_{\mathcal{H}_g}(v_1,v_3,v_4,v_2)$, we split $P$ into three sub-paths, $P_1$, $P_2$ and $P_3$, such that $P_1 \in S_{\mathcal{H}_g}(v_1,v_3)$, $P_2 \in S_{\mathcal{H}_g}(v_3,v_4)$ and $P_3 \in S_{\mathcal{H}_g}(v_4,v_2)$. Notice that $P_1$ corresponds to a nice path of $\mathcal{H}_{g-1}^{(1)}$ from $v_1^{(1)}$ to  $v_2^{(1)}$. By induction hypothesis, $P_1$ is oddly oriented. Analogously, we can prove that $P_2$ and $P_3$ are both oddly oriented. Therefore, $P$ is oddly oriented.
\end{proof}

\begin{refproof}[Proof of Theorem~\ref{Lem:LW04}.]
In order to prove that $\mathcal{H}_g^e$ is a Pfaffian orientation of $\mathcal{H}_g$, we only require to prove that every nice cycle of $\mathcal{H}_g$ is oddly oriented relative to the orientation $\mathcal{H}_g^e$ of $\mathcal{H}_g$. By induction. For $g=1$, $\mathcal{H}_1$ has a unique nice cycle. It is easy to see that this nice cycle is oddly oriented relative to $\mathcal{H}_1^e$.

For $g>1$, assume that the statement is true for all $\mathcal{H}_j$ ($1 \leq j < g$). Let $C$ be an arbitrary nice cycle of $\mathcal{H}_g$. By construction, $C$  belongs to either a subgraph $\mathcal{H}_{g-1}^{(i)}$, $i=1,2,3,4$, of $\mathcal{H}_g$ or set $S_{\mathcal{H}_g}(v_1,v_3,v_4,v_2,v_1)$.

When $C$ belongs to $\mathcal{}H_{g-1}^{(i)}$, $i=1,2,3,4$, we can prove that there exists a subgraph $\mathcal{L}$ of $\mathcal{H}_{g-1}^{(i)}$ satisfying that $\mathcal{L}$ is isomorphic to $\mathcal{H}_k$ (with $k$ being the smallest integer between $1$ and $g-1$) and ${C}$ is a nice cycle of $\mathcal{L}$. Such a subgraph $\mathcal{L}$ can be obtained in the following manner. First, let $\mathcal{L}=\mathcal{H}_{g-1}^{(i)}$. If ${C}$ belongs to one of the four mutually isomorphic subgraphs $\mathcal{H}_{g-2}^{(i')}$, $i'=1,2,3,4$, forming $\mathcal{L}$, then let $\mathcal{L}=H_{g-2}^{(i')}$. In an analogous  way, by iteratively using the operations on $\mathcal{H}_{g-2}^{(i')}$ and the resulting subgraphs, we can find the smallest integer $k$ ($1 \leq k < g$), such that $\mathcal{L}$ is isomorphic to $\mathcal{H}_k$. We next show that $C$ is a nice cycle of $\mathcal{L}$. {Let $v_1^{*}$, $v_2^{*}$, $v_3^{*}$ and $v_4^{*}$ be the four hub vertices of $\mathcal{L}$, corresponding to the hubs $v_1$, $v_2$, $v_3$ and $v_4$ in $\mathcal{H}_k$. Then $C$ must belongs to $S_{\mathcal{L}}(v_1^{*}, v_3^{*}, v_4^{*}, v_2^{*}, v_1^{*})$. Therefore, in $\mathcal{H}_g \setminus C$, the vertices in $\mathcal{L} \setminus C$ are separated from other ones.  Because $C$ is a nice cycle of $\mathcal{H}_g$, $\mathcal{H}_g \setminus C$ has a perfect matching, implying that $\mathcal{L} \setminus C$ has also a perfect matching. Hence, $C$ is a nice cycle of $\mathcal{L}$. By induction hypothesis, $C$ is oddly oriented relative to $\mathcal{L}$, indicating that $C$ is oddly oriented with respect to  $\mathcal{H}_g^e$.}

When the nice cycle $C \in S_{\mathcal{H}_g}(v_1,v_3,v_4,v_2,v_1)$, it can be split into four nice sub-paths $P_1$, $P_2$, $P_3$ and $P_4$, with $P_1 \in S_{\mathcal{H}_g}(v_1,v_3)$, $P_2 \in S_{\mathcal{H}_g}(v_3,v_4)$, $P_3 \in S_{\mathcal{H}_g}(v_4,v_2)$, and $P_4 \in S_{\mathcal{H}_g}(v_2,v_1)$. {Then, $P_1$, $P_2$, $P_3$ are nice paths of $\mathcal{H}^{(1)}_{g-1}$, $\mathcal{H}^{(2)}_{g-1}$, and $\mathcal{H}^{(3)}_{g-1}$, respectively. By Lemma~\ref{Lem:LW03}, $P_1$, $P_2$, $P_3$  are all oddly oriented. Analogously, $P_4$ is evenly oriented. Thus,  $C$ is oddly oriented with respect to $\mathcal{H}_g^e$.}
\end{refproof}

\subsection{Number of perfect matchings}

We are now ready to determine the number and entropy of perfect matchings in network $\mathcal{H}_g$. The main results can be stated as follows.

\begin{theo}\label{Lem:LW12}
The number of perfect matchings of $\mathcal{H}_g$, $g\geq 1$, is $\psi(H_g)=2^{\frac{1}{9}\cdot4^g+\frac{2}{3}g-\frac{1}{9}}$, and the entropy for perfect matchings in $\mathcal{H}_g$, $g\to \infty$, is $z(\mathcal{H}_g)=\frac{\ln2}{3}$.
\end{theo}


Below we will prove Theorem \ref{Lem:LW12} by evaluating the determinant of the skew adjacency matrix for the Pfaffian orientation $\mathcal{H}_g^e$ of network $\mathcal{H}_g$. To this end, we first introduce some additional quantities and provide some lemmas.

\begin{definition}\label{Def:LW03}
The six matrices $A_g$, $B_g$, $B_g^{\prime}$, $D_g$, $D_g^{\prime}$ and $K_g$ associated with the Pfaffian orientation $\mathcal{H}_g^e$ are defined as follows:

$A_g$ is the skew adjacency matrix $A(H_g^e)$, for simplicity.

$B_g$ (or $B_g^{\prime}$) is a sub-matrix of $A_g$, obtained by deleting from $A_g$ the row and column corresponding to vertex $v_1$ (or $v_2$).

$D_g$ (or $D_g^{\prime}$) is a sub-matrix of $A_g$, obtained by deleting from $A_g$ the row corresponding to vertex $v_1$ (or $v_2$) and the column corresponding to vertex $v_2$ (or $v_1$).

$K_g$ is a sub-matrix of $A_g$, obtained by deleting from $A_g$ two rows and two columns corresponding to vertex $v_1$ and $v_2$.
\end{definition}

The following Lemma is immediate from the second construction of network $\mathcal{H}_{g+1}$, see Definition~\ref{Def:NF02}.
\begin{lemma}\label{Lem:LW05}
For $g\geq 1$, matrices $A_{g+1}$, $B_{g+1}$, $B_{g+1}^{\prime}$, $D_{g+1}$,  $D_{g+1}^{\prime}$ and $K_{g+1}$ satisfy the following relations:
\begin{equation}\label{Equ:LW01}
A_{g+1}=\left(
\begin{array}{cccccccc}
 0 & 1 & 1 & 0 & x_g & o  & o  & x_g \\
 -1 & 0 & 0 & -1 & o  & o & y_g & y_g  \\
 -1 & 0 & 0 & 1 & y_g & x_g & o  & o  \\
 0 & 1 & -1 & 0 & o  & y_g  & x_g & o \\
 -x_g^{\top} & o^{\top}  & -y_g^{\top} & o^{\top}  & K_g & O  & O  & O  \\
 o^{\top}  & o^{\top} & -x_g^{\top} & -y_g^{\top}  & O  & K_g & O  & O  \\
 o^{\top}  & -y_g^{\top} & o^{\top}  & -x_g^{\top} & O  & O  & K_g & O  \\
 -x_g^{\top} & -y_g^{\top}  & o^{\top}  & o^{\top} & O  & O  & O  & K_g
\end{array}\right)\,,
\end{equation}
\begin{equation}\label{Equ:LW02}
B_{g+1}=\left(
\begin{array}{ccccccc}
 0 & 0 & -1 & o  & o & y_g & y_g  \\
 0 & 0 & 1 & y_g & x_g & o  & o  \\
 1 & -1 & 0 & o  & y_g  & x_g & o \\
 o^{\top}  & -y_g^{\top} & o^{\top}  & K_g & O  & O  & O  \\
 o^{\top} & -x_g^{\top} & -y_g^{\top}  & O  & K_g & O  & O  \\
 -y_g^{\top} & o^{\top}  & -x_g^{\top} & O  & O  & K_g & O  \\
 -y_g^{\top}  & o^{\top}  & o^{\top} & O  & O  & O  & K_g
\end{array}\right)\,,
\end{equation}
\begin{equation}\label{Equ:LW03}
B_{g+1}^{\prime}=\left(
\begin{array}{ccccccc}
 0 & 1 & 0 & x_g & o  & o  & x_g \\
 -1 & 0 & 1 & y_g & x_g & o  & o  \\
 0  & -1 & 0 & o  & y_g  & x_g & o \\
 -x_g^{\top}   & -y_g^{\top} & o^{\top}  & K_g & O  & O  & O  \\
 o^{\top}   & -x_g^{\top} & -y_g^{\top}  & O  & K_g & O  & O  \\
 o^{\top}   & o^{\top}  & -x_g^{\top} & O  & O  & K_g & O  \\
 -x_g^{\top}   & o^{\top}  & o^{\top} & O  & O  & O  & K_g
\end{array}\right)\,,
\end{equation}
\begin{equation}\label{Equ:LW04}
D_{g+1}=\left(
\begin{array}{ccccccc}
 -1 & 0 & -1 & o  & o & y_g & y_g  \\
 -1 & 0 & 1 & y_g & x_g & o  & o  \\
 0 & -1 & 0 & o  & y_g  & x_g & o \\
 -x_g^{\top}   & -y_g^{\top} & o^{\top}  & K_g & O  & O  & O  \\
 o^{\top}   & -x_g^{\top} & -y_g^{\top}  & O  & K_g & O  & O  \\
 o^{\top}   & o^{\top}  & -x_g^{\top} & O  & O  & K_g & O  \\
 -x_g^{\top} & o^{\top}  & o^{\top} & O  & O  & O  & K_g
\end{array}\right)\,,
\end{equation}
\begin{equation}\label{Equ:LW05}
D_{g+1}^{\prime}=\left(
\begin{array}{ccccccc}
 1 & 1 & 0 & x_g & o  & o  & x_g \\
 0 & 0 & 1 & y_g & x_g & o  & o  \\
 1 & -1 & 0 & o  & y_g  & x_g & o \\
 o^{\top}  & -y_g^{\top} & o^{\top}  & K_g & O  & O  & O  \\
 o^{\top} & -x_g^{\top} & -y_g^{\top}  & O  & K_g & O  & O  \\
 -y_g^{\top} & o^{\top}  & -x_g^{\top} & O  & O  & K_g & O  \\
 -y_g^{\top}  & o^{\top}  & o^{\top} & O  & O  & O  & K_g
\end{array}\right)\,,
\end{equation}
and
\begin{equation}\label{Equ:LW06}
K_{g+1}=\left(
\begin{array}{cccccc}
 0 & 1 & y_g & x_g & o  & o  \\
 -1 & 0 & o  & y_g  & x_g & o \\
 -y_g^{\top} & o^{\top}  & K_g & O  & O  & O  \\
 -x_g^{\top} & -y_g^{\top}  & O  & K_g & O  & O  \\
 o^{\top}  & -x_g^{\top} & O  & O  & K_g & O  \\
 o^{\top}  & o^{\top} & O  & O  & O  & K_g
\end{array}\right)\,,
\end{equation}
where $x_g$ and $y_g$ are two $(N_g-2)$-dimensional row vectors describing, respectively, the adjacency relation between the two hub vertices $v_1$ and $v_2$, and other vertices in network $\mathcal{H}_g$; $O$ (or $o$) is zero matrix (or zero vector) of appropriate order; and the superscript $\top$ of a vector represents transpose.
\end{lemma}

Equation~(\ref{Equ:LW01}) can be accounted for as follows. Let us represent $A_{g+1}$ in the block form: $A_{g+1}^{(i,j)}$ ($i,j=1,2,\ldots,8$) denote the block of $A_{g+1}$ at row $i$ and column $j$. Let $\overline{\mathcal{V}}_{g}^{(i)}$ be the vertex set of $\mathcal{H}_g^{(i)}$, with the two hub vertices corresponding to $v_1$ and $v_2$ in $\mathcal{H}_g$ being  removed, see Fig.~\ref{Fig.3}. Then $A_{g+1}^{(i,j)}$ ($i,j=1,2,3,4$) represents the adjacency relation between vertex $v_i$ and $v_j$. Similarly, $A_{g+1}^{(i,j)}$ ($i,j=5,6,7,8$) represents the adjacency relation between vertices in set $\overline{\mathcal{V}}_{g}^{(i-4)}$ and vertices in set $\overline{\mathcal{V}}_{g}^{(j-4)}$. Fig.~\ref{Fig.3} shows when $i\neq j$, there exists no edge between vertices in $\overline{\mathcal{V}}_{g}^{(i-4)}$ and vertices in $\overline{\mathcal{V}}_{g}^{(j-4)}$, so the corresponding block $A_{g+1}^{(i,j)}=O$.
Equations~(\ref{Equ:LW02}),~(\ref{Equ:LW03}),~(\ref{Equ:LW04}),~(\ref{Equ:LW05}),  and~(\ref{Equ:LW06}) can be accounted for in an analogous way.

The following lemmas are useful for determining the number of perfect matching in network $\mathcal{H}_g$.
\begin{lemma}\label{Lem:LW06}
For $g\geq 1$, $\det (B_g )= \det (B_g^{\prime} )=0$.
\end{lemma}
\begin{proof}
By definition, $B_g$ is an antisymmetric matrix and its order is odd order. Then,
\begin{equation}\label{Equ:LW06X}
\det (B_g )=\det (B_g^\top )=(-1)^{N_g -1}\det (B_g )=-\det (B_g )\,,
\end{equation}
which yields $\det (B_g )=0$. Similarly, we can prove  $\det (B_g^{\prime} )=0$.
\end{proof}


\begin{lemma}\label{Lem:LW07}
For $g\geq 1$, $\det(D_g^{\prime})=-\det (D_g)$.
\end{lemma}
\begin{proof}
From Eqs.~(\ref{Equ:LW04}) and~(\ref{Equ:LW05}), one obtains
\begin{equation}\label{Equ:LW06Y}
D_{g+1}^{\prime}=-D_{g+1}^{\top}\,.
\end{equation}
Then,
\begin{equation}\label{Equ:LW06Z}
\det(D_{g+1}^{\prime})=(-1)^{N_{g+1}-1}\det(D_{g+1})=-\det(D_{g+1})\,,
\end{equation}
which, together with $\det(D_1^{\prime})=-\det(D_1)$, results in $\det(D_g^{\prime})=-\det (D_g)$.
\end{proof}

\begin{lemma}\label{Lem:LW08}
For $g\geq 1$, $ \det(K_{g+1})= [\det(K_g)]^3\det(A_g)$.
\end{lemma}
\begin{proof}
By using Laplace's theorem~\cite{St09} to Eq.~(\ref{Equ:LW06}), we obtain
\begin{align}\label{Equ:LW07}
\det( K_{g+1} )&=
\det \left(
\begin{array}{cccccc}
 0 & 1 & o & o & o  & o  \\
 -1 & 0 & o  & y_g  & x_g & o \\
 -y_g^{\top} & o^{\top}  & K_g & O  & O  & O  \\
 -x_g^{\top} & -y_g^{\top}  & O  & K_g & O  & O  \\
 o^{\top}  & -x_g^{\top} & O  & O  & K_g & O  \\
 o^{\top}  & o^{\top} & O  & O  & O  & K_g
\end{array}
\right)+\nonumber\\
&\quad
\det \left(
\begin{array}{cccccc}
 0 & 0 & y_g & o & o  & o  \\
 -1 & 0 & o  & y_g  & x_g & o \\
 -y_g^{\top} & o^{\top}  & K_g & O  & O  & O  \\
 -x_g^{\top} & -y_g^{\top}  & O  & K_g & O  & O  \\
 o^{\top}  & -x_g^{\top} & O  & O  & K_g & O  \\
 o^{\top}  & o^{\top} & O  & O  & O  & K_g
\end{array}
\right)+\nonumber\\
&\quad
\det \left(
\begin{array}{cccccc}
 0 & 0 & o & x_g & o  & o  \\
 -1 & 0 & o  & y_g  & x_g & o \\
 -y_g^{\top} & o^{\top}  & K_g & O  & O  & O  \\
 -x_g^{\top} & -y_g^{\top}  & O  & K_g & O  & O  \\
 o^{\top}  & -x_g^{\top} & O  & O  & K_g & O  \\
 o^{\top}  & o^{\top} & O  & O  & O  & K_g
\end{array}
\right)
\end{align}

Let $\Theta_{g+1}^{(i)}$, $1 \leq i \leq 3$, denote sequentially the three determinants on the rhs of Eq.~(\ref{Equ:LW07}). Applying some elementary matrix operations {and the properties of determinants}, we  obtain
\begin{align}\label{Equ:LW08}
\Theta_{g+1}^{(1)} &=
\det (K_g) \,
\det\left(
\begin{array}{cccccc}
 0 & 1 & o & o & o  \\
 -1 & 0 & o  & y_g  & x_g \\
 -y_g^{\top} & o^{\top}  & K_g & O  & O \\
 -x_g^{\top} & -y_g^{\top}  & O  & K_g & O  \\
 o^{\top}  & -x_g^{\top} & O  & O  & K_g
\end{array}
\right)
\nonumber \\
&= -\, \det (K_g) \,
\det \left(
\begin{array}{cccccc}
 1 & o  \\
-x_g^{\top} & K_g
\end{array}
\right)
\,
\det \left(
\begin{array}{cccccc}
 -1 & o  & y_g\\
 -y_g^{\top} &K_g & O\\
 -x_g^{\top} &O  & K_g
\end{array}
\right)
\,
\nonumber \\
&= -\, [\det (K_g)]^2 \,
\det \left(
\begin{array}{cccccc}
 -1 & o  & y_g\\
 -y_g^{\top} &K_g & O\\
 -x_g^{\top} &O  & K_g
\end{array}
\right)
\nonumber \\
&= -\, [\det (K_g)]^3 \,
\det \left(
\begin{array}{cccccc}
 -1 & y_g\\
 -x_g^{\top} & K_g
\end{array}
\right)\,,
\end{align}
\begin{align}\label{Equ:LW09}
\Theta_{g+1}^{(2)} &= \det (K_g) \, \det \left(
\begin{array}{cccccc}
 0 & 0 & y_g  & o  & o  \\
 -1 & 0 & o    & y_g & x_g \\
 -y_g^{\top} & o^{\top}  & K_g  & O  & O  \\
 -x_g^{\top}  & -y_g^{\top} & O    & K_g & O  \\
 o^{\top}  & -x_g^{\top} & O    & O  & K_g
\end{array}\right)
\nonumber \\
&= \det (K_g) \, \det \left(
\begin{array}{cc}
0 & y_g\\
-y_g^{\top} & K_g
\end{array}\right) \, \det \left(
\begin{array}{ccc}
 0 &  y_g & x_g \\
 -y_g^{\top} &  K_g & O  \\
 -x_g^{\top} &  O  & K_g
\end{array}\right) \,,
\end{align}
and
\begin{align}\label{Equ:LW10}
\Theta_{g+1}^{(3)}
&= [\det (K_g)]^2 \det \left(
\begin{array}{cccccc}
 0 & 0 & x_g & o  \\
 -1 & 0 & y_g  & x_g \\
 -x_g^{\top} & -y_g^{\top}  & K_g & O  \\
 o^{\top}  & -x_g^{\top} & O  & K_g
\end{array}\right) \nonumber \\
&= [\det (K_g)]^2 \det \left(
\begin{array}{cccccc}
 0 & 0 & x_g & o  \\
 -1 & 0 & y_g  & o \\
 -x_g^{\top} & -y_g^{\top}  & K_g & O  \\
 o^{\top}  & 0 & O  & K_g
\end{array}\right) + \nonumber \\
&\quad [\det (K_g)]^2 \det \left(
\begin{array}{cccccc}
 0 & 0 & x_g & o  \\
 0 & 0 & o & x_g \\
 -x_g^{\top} & 0 & K_g & O  \\
 o^{\top}  & -x_g^{\top} & O  & K_g
\end{array}\right)\nonumber \\
&= [\det (K_g)]^3 \det \left(
\begin{array}{cccccc}
 0 & 0 & x_g\\
 -1 & 0 & y_g\\
 -x_g^{\top} & -y_g^{\top} & K_g
\end{array}\right) +\nonumber \\
&\quad
[\det (K_g)]^2 \det \left(
\begin{array}{cc}
 0 & x_g  \\
 -x_g^{\top} & K_g
\end{array}\right)
\det \left(
\begin{array}{cc}
 0 & x_g \\
 -x_g^{\top} & K_g
\end{array}\right)\nonumber \\
&= [\det (K_g)]^3 \det \left(
\begin{array}{cccccc}
 0 & 1 & x_g\\
 -1 & 0 & y_g\\
 -x_g^{\top} & -y_g^{\top}  & K_g
\end{array}\right)+ [\det (K_g)]^3 \det \left(
\begin{array}{cccccc}
 -1 & y_g\\
 -x_g^{\top} & K_g
\end{array}\right) \nonumber \\
&\quad +
[\det (K_g)]^2 \det \left(
\begin{array}{cc}
 0 & x_g  \\
 -x_g^{\top} & K_g
\end{array}\right)
\det \left(
\begin{array}{cc}
 0 & x_g \\
 -x_g^{\top} & K_g
\end{array}\right) \, . \nonumber
\end{align}

Note that both matrices $\left(
\begin{array}{cc}
0 & y_g\\
-y_g^{\top} & K_g
\end{array}\right)$ and $\left(
\begin{array}{cc}
0 & x_g\\
-x_g^{\top} & K_g
\end{array}\right)$ are antisymmetric and have odd order, which implies
\begin{equation}\label{Equ:LW11}
\det \left(
\begin{array}{cc}
0 & y_g\\
-y_g^{\top} & K_g
\end{array}\right) = \det \left(
\begin{array}{cc}
0 & x_g\\
-x_g^{\top} & K_g
\end{array}\right) = 0\,. \nonumber
\end{equation}
By Definition \ref{Def:LW03}, we have
\begin{equation}
\det \left(
\begin{array}{cccccc}
 0 & 1 & x_g\\
 -1 & 0 & y_g\\
 -x_g^{\top} & -y_g^{\top}  & K_g
\end{array}\right) =\det (A_g)\nonumber
\end{equation}
and
\begin{equation}
\det \left(
\begin{array}{cccccc}
 -1 & y_g\\
 -x_g^{\top} & K_g
\end{array}\right) = \det (D_g)\, .\nonumber
\end{equation}
Thus,
\begin{equation}\label{Equ:LW12A}
\Theta_{g+1}^{(1)} = -[\det (K_g)]^3\det(D_g)\,, \nonumber
\end{equation}
\begin{equation}\label{Equ:LW12B}
\Theta_{g+1}^{(2)} = 0\,,\nonumber
\end{equation}
\begin{equation}\label{Equ:LW12C}
\Theta_{g+1}^{(3)} = [\det (K_g)]^3 \det (A_g)+[\det (K_g)]^3 \det(D_g),\nonumber
\end{equation}
which leads to
\begin{equation}\label{Equ:LW12}
\det (K_{g+1}) = \Theta_{g+1}^{(1)} + \Theta_{g+1}^{(2)} + \Theta_{g+1}^{(3)} = [\det (K_g)]^3 \det (A_g). \nonumber
\end{equation}
This completes the proof of the lemma.
\end{proof}


\begin{lemma}\label{Lem:LW09}
For $g\geq 1$, $\det (A_{g+1})=4[\det(A_g)]^2[\det(K_g)]^2$.
\end{lemma}
\begin{proof}
By using Laplace's theorem~\cite{St09} to Eq.~(\ref{Equ:LW01}), we have
\begin{align*}
\det (A_{g+1}) &=\det\left(
\begin{array}{cccccccc}
 0 & 1 & 0 & 0 & o & o  & o  & o \\
 -1 & 0 & 0 & -1 & o  & o & y_g & y_g  \\
 -1 & 0 & 0 & 1 & y_g & x_g & o  & o  \\
 0 & 1 & -1 & 0 & o  & y_g  & x_g & o \\
 -x_g^{\top} & o^{\top}  & -y_g^{\top} & o^{\top}  & K_g & O  & O  & O  \\
 o^{\top}  & o^{\top} & -x_g^{\top} & -y_g^{\top}  & O  & K_g & O  & O  \\
 o^{\top}  & -y_g^{\top} & o^{\top}  & -x_g^{\top} & O  & O  & K_g & O  \\
 -x_g^{\top} & -y_g^{\top}  & o^{\top}  & o^{\top} & O  & O  & O  & K_g
\end{array}\right)\, \notag\\
&+ \det\left(
\begin{array}{cccccccc}
 0 & 0 & 1 & 0 & o & o  & o  & o \\
 -1 & 0 & 0 & -1 & o  & o & y_g & y_g  \\
 -1 & 0 & 0 & 1 & y_g & x_g & o  & o  \\
 0 & 1 & -1 & 0 & o  & y_g  & x_g & o \\
 -x_g^{\top} & o^{\top}  & -y_g^{\top} & o^{\top}  & K_g & O  & O  & O  \\
 o^{\top}  & o^{\top} & -x_g^{\top} & -y_g^{\top}  & O  & K_g & O  & O  \\
 o^{\top}  & -y_g^{\top} & o^{\top}  & -x_g^{\top} & O  & O  & K_g & O  \\
 -x_g^{\top} & -y_g^{\top}  & o^{\top}  & o^{\top} & O  & O  & O  & K_g
\end{array}\right)\, \notag\\
&+ \det\left(
\begin{array}{cccccccc}
 0 & 0 & 0 & 0 & x_g & o  & o  & o \\
 -1 & 0 & 0 & -1 & o  & o & y_g & y_g  \\
 -1 & 0 & 0 & 1 & y_g & x_g & o  & o  \\
 0 & 1 & -1 & 0 & o  & y_g  & x_g & o \\
 -x_g^{\top} & o^{\top}  & -y_g^{\top} & o^{\top}  & K_g & O  & O  & O  \\
 o^{\top}  & o^{\top} & -x_g^{\top} & -y_g^{\top}  & O  & K_g & O  & O  \\
 o^{\top}  & -y_g^{\top} & o^{\top}  & -x_g^{\top} & O  & O  & K_g & O  \\
 -x_g^{\top} & -y_g^{\top}  & o^{\top}  & o^{\top} & O  & O  & O  & K_g
\end{array}\right)\, \notag\\
&+ \det\left(
\begin{array}{cccccccc}
 0 & 0 & 0 & 0 & o & o  & o  & x_g \\
 -1 & 0 & 0 & -1 & o  & o & y_g & y_g  \\
 -1 & 0 & 0 & 1 & y_g & x_g & o  & o  \\
 0 & 1 & -1 & 0 & o  & y_g  & x_g & o \\
 -x_g^{\top} & o^{\top}  & -y_g^{\top} & o^{\top}  & K_g & O  & O  & O  \\
 o^{\top}  & o^{\top} & -x_g^{\top} & -y_g^{\top}  & O  & K_g & O  & O  \\
 o^{\top}  & -y_g^{\top} & o^{\top}  & -x_g^{\top} & O  & O  & K_g & O  \\
 -x_g^{\top} & -y_g^{\top}  & o^{\top}  & o^{\top} & O  & O  & O  & K_g
\end{array}\right)\,.
\end{align*}

We use $\Lambda_{g+1}^{(i)}$  ($i=1,2,3,4$) to denote  sequentailly the four  determinants on the rhs of the above equation.  As in the proof of Lemma~\ref{Lem:LW08}, by applying some elementary matrix operations, we have
\begin{align*}
\Lambda_{g+1}^{(1)} = \Lambda_{g+1}^{(2)}
= -2 \det(A_g) \det(D_g) \left[\det(K_g)\right]^2\,,
\end{align*}
\begin{align*}
\Lambda_{g+1}^{(3)} = \Lambda_{g+1}^{(4)}
= 2\left[\det(A_g)\right]^2 \left[\det(K_g)\right]^2 + 2\det(A_g) \det(D_g) \left[\det(K_g)\right]^2\, ,
\end{align*}
which leads to
\begin{align*}
\det (A_{g+1}) = \Lambda_{g+1}^{(1)} + \Lambda_{g+1}^{(2)} + \Lambda_{g+1}^{(3)}
+ \Lambda_{g+1}^{(4)} = 4\left[\det(A_g)\right]^2 \left[\det(K_g)\right]^2 \, ,
\end{align*}
as desired.
\end{proof}



\begin{lemma}\label{Lem:LW11}
For $g\geq 1$, $\det(A_g)=4^{\frac{1}{9}\cdot4^g+\frac{2}{3}g-\frac{1}{9}}$.
\end{lemma}
\begin{proof}
By Lemmas~\ref{Lem:LW08} and \ref{Lem:LW09}, the result follows by considering  initial conditions $\det (A_1)=4$ and $\det (K_1)=2$.
\end{proof}
\begin{refproof}[Proof of Theorem~\ref{Lem:LW12}.]
From Lemmas~\ref{Lem:Pre01} and~\ref{Lem:LW11}, the number of perfect matchings in network $\mathcal{H}_g$ is
\begin{equation}
\psi(\mathcal{H}_g)=\sqrt{\det(A_g)}=2^{\frac{1}{9}\cdot4^g+\frac{2}{3}g-\frac{1}{9}}\,,\nonumber
\end{equation}
and the entropy of perfect matchings in network $\mathcal{H}_g$ ($g \to \infty$) is
\begin{equation}\label{Equ:LW14}
z(\mathcal{H}_g)=\lim\limits_{g\to \infty} \frac{\ln \psi(\mathcal{H}_g)}{\frac{N_g}{2}}=\lim\limits_{g\to \infty} \frac{\ln 2^{\frac{1}{9}\cdot4^g+\frac{2}{3}g-\frac{1}{9}}}{\frac{4^g + 2}{3}}=\frac{\ln2}{3}\,. \nonumber
\end{equation}
This completes the proof of Theorem~\ref{Lem:LW12}.
\end{refproof}



\subsection{Comparison with the extended Sierpi\'nski graph}

We have shown that for the two networks $\mathcal{F}_g$  and  $\mathcal{H}_g$ with the same degree sequence, the fractal scale-free network $\mathcal{F}_g$ has no perfect matchings; in sharp contrast, the non-fractal scale-free network $\mathcal{H}_g$ has perfect matchings. Moreover, the number of perfect matchings in $\mathcal{H}_g$ is very high, the entropy of which is equivalent to that of the extended Sierpi\'nski graph~\cite{KlMo05}, as will be shown below.  




\begin{figure}
\centering
\includegraphics[width=0.4\textwidth]{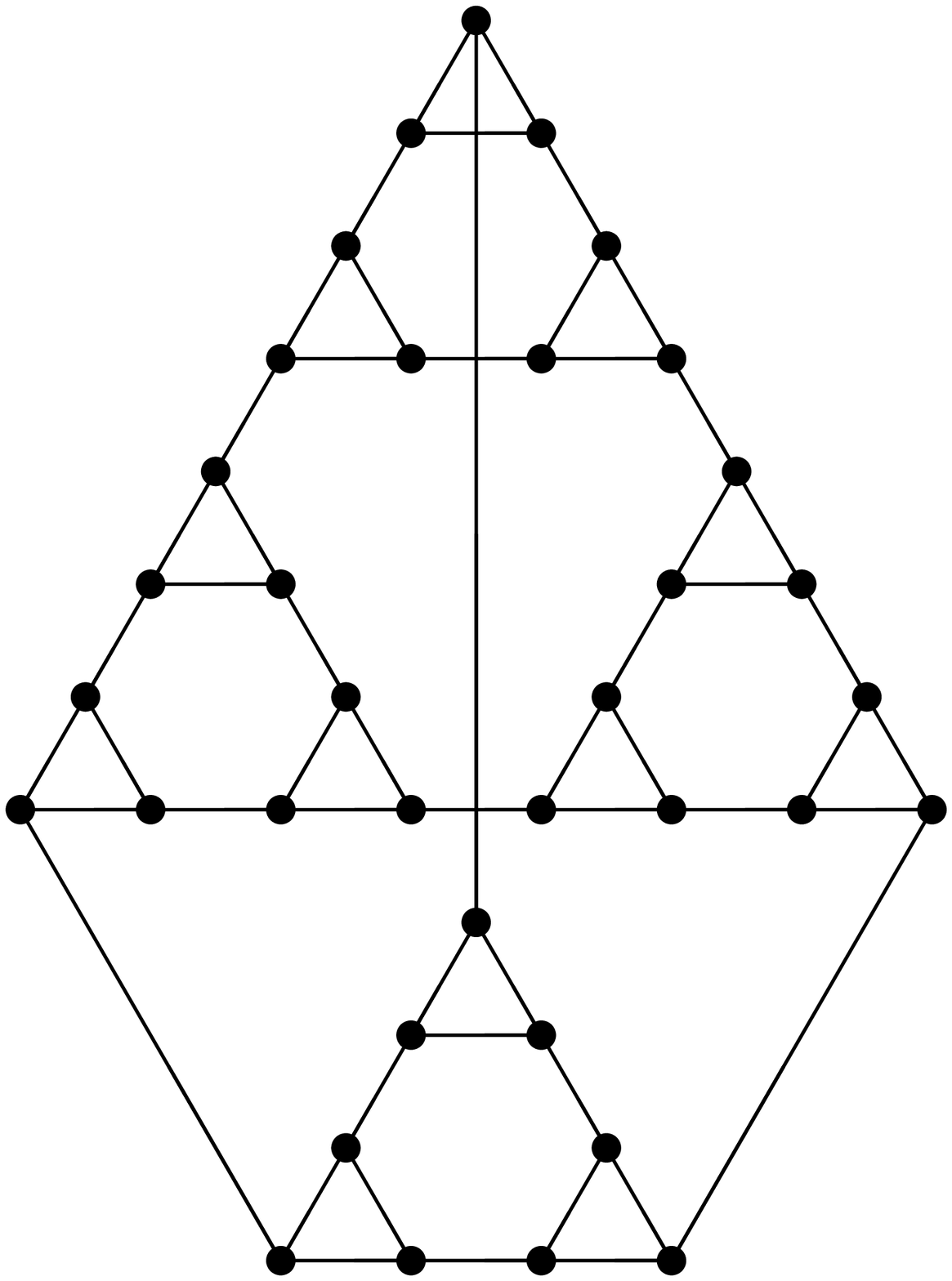}
\caption{The extended Sierpi\'nski graph $\mathcal{S}^{++}_3$.}
\label{spp3}
\end{figure}




The extended Sierpi\'nski graph is a particular case of {S}ierpi{\'n}ski-like graphs  proposed by Klav{\v{z}}ar and Moharin~\cite{KlMo05}, which is in fact a variant of the Tower of Hanoi graph~\cite{HiKlMiPeSt13,ZhWuLiCo16}. The extended Sierpi\'nski graph can be  defined by iteratively applying the subdivided-line
graph operation~\cite{Ha15}. Denote by $\Gamma^1 (\mathcal{G }) = L(B(\mathcal{G})) $, and the $g$th subdivided-line of $\mathcal{G}$ is obtained through the iteration $\Gamma^g(\mathcal{G }) = \Gamma(\Gamma^{g-1}(\mathcal{G }))$. Let $\mathcal{K}_4$ be the complete graph of 4 vertices, and let $\mathcal{S}^{++}_g$ denote the extended Sierpi\'nski graph. Then $\mathcal{S}^{++}_g$, $g\geq 1$, is defined by  $\mathcal{S}^{++}_g=\Gamma^{g-1}(\mathcal{K}_4)$, with $\mathcal{S}^{++}_1=\mathcal{K}_4$. Fig.~\ref{spp3} illustrates an extended Sierpi\'nski graph $\mathcal{S}^{++}_3$.

For all $g\geq 1$, the extended Sierpi\'nski graph $\mathcal{S}^{++}_g$ is a $3$-regular graph. Moreover,  $\mathcal{S}^{++}_g$ is fractal but not small-world.
By definition, it is easy to verify that the number of vertices and edges in the extended Sierpi\'nski graph $\mathcal{S}^{++}_g$ are $N_g = 4\cdot 3^{g-1}$ and $E_g = 2\cdot 3^g$, respectively.








\begin{theo}
The number of perfect matchings in the extended Sierpi\'nski graphs $\mathcal{S}^{++}_g$, $g\geq 1$,
is $\psi(\mathcal{S}^{++}_g)=2^{2\cdot 3^{g-2} + 1}$, and the entropy for perfect matchings in $\mathcal{S}^{++}_g$, $g\to \infty$, is $z(\mathcal{S}^{++}_g)=\frac{\ln2}{3}$.
\end{theo}
\begin{proof}
By definition, $\mathcal{S}^{++}_g=L(B(\mathcal{S}^{++}_{g-1}))$. It is obvious that $B(\mathcal{S}^{++}_{g-1})$ has $N_{g-1}+E_{g-1}=4\cdot 3^{g-2}+2\cdot 3^{g-1}$ vertices
and $2E_{g-1}=4\cdot 3^{g-1}$ edges. Moreover, the degree of a vertex in $B(\mathcal{S}^{++}_{g-1})$ is either 2 or 3. From Lemma~\ref{linegraph}, for all $g\geq 1$, the number of perfect matchings in $\mathcal{S}^{++}_g$ is
\begin{equation*}
\psi(\mathcal{S}^{++}_g) = \psi(L(B(\mathcal{S}^{++}_{g-1}))) = 2^{2E_{g-1}
- (N_{g-1}+ E_{g-1}) + 1} = 2^{2\cdot 3^{g-2} + 1}.
\end{equation*}
Then, and the entropy of perfect matchings in  extended Sierpi\'nski graphs $\mathcal{S}^{++}_g$, $g \to \infty$, is
\begin{equation*}
z(\mathcal{S}^{++}_g)=\lim\limits_{g\to \infty} \frac{\ln \psi(\mathcal{S}^{++}_g)}{\frac{N_g}{2}}=\lim\limits_{g\to \infty} \frac{\ln 2^{2\cdot 3^{g-2} + 1}}{2\cdot 3^{g-1}}=\frac{\ln2}{3}\,,
\end{equation*}
as the theorem claims.
\end{proof}

\section{Conclusion}

In this paper, we have studied both the size and the number of maximum matchings in two self-similar scale-free networks with identical degree distribution, and shown that the first network has no perfect matchings, while the second network has many perfect matchings. For the first network, we determined explicitly the size and number of maximum matchings by using its self-similarity. For the second network, we constructed a Pfaffian orientation, using the skew adjacency matrix of which we determined the exact number of perfect matchings and its associated entropy. Furthermore, we determined the number of perfect matchings in an extended regular Sierpi\'nski graph, and demonstrated that entropy for its perfect matchings equals that of the second scale-free network. Thus, power-law degree distribution itself is not enough to characterize maximum matchings in scale-free networks, and care should be needed when making a general statement on maximum matchings in scale-free networks. Due to the relevance of maximum matchings to structural controllability, our work is helpful for better understanding controllability of scale-free networks.

\section*{Acknowledgements}

This work is supported by the National Natural Science Foundation of China under Grant No. 11275049. 

\end{document}